\documentclass[showpacs,amsmath,amssymb,onecolumn,pra,superscriptaddress,nofootinbib]{revtex4-2}

\usepackage{hyperref}
\usepackage{qcircuit}
\usepackage[dvips]{graphicx}
\usepackage{amsmath,amssymb,amsthm,mathrsfs,amsfonts,dsfont}
\usepackage{braket}
\usepackage{bm}
\usepackage{enumerate}
\usepackage{color}
\usepackage{graphicx}
\usepackage{mathrsfs}
\usepackage{multirow}
\usepackage{algorithm}
\usepackage{algorithmic}
\usepackage{tcolorbox}
\usepackage{bbm}
\usepackage{float}
\usepackage{enumitem}
\usepackage{diagbox}
\usepackage{multirow}
\usepackage{newfloat}
\usepackage{natbib}
\setcitestyle{number}
\usepackage[normalem]{ulem}  
\usepackage{cancel}
\usepackage{xcolor}

\hypersetup{
    colorlinks=true,
    linkcolor=blue,
    filecolor=magenta,      
    urlcolor=cyan,
    pdftitle={Overleaf Example},
    pdfpagemode=FullScreen,
    }

\begin{document}

\preprint{APS/123-QED}

\title{Randomization Accelerates Series-Truncated Quantum Algorithms}

\author{Yue Wang}
\author{Qi Zhao}%
\email{Contact author: zhaoqi@cs.hku.hk}
\affiliation{%
QICI Quantum Information and Computation Initiative, School of Computing and Data Science,
The University of Hong Kong, Pokfulam Road, Hong Kong SAR, China
}

\date{February 8, 2024}

\begin{abstract}


Quantum algorithms typically demand prohibitively complicated circuits to solve practical problems. Previous studies have shown that classical randomness can accelerate some specific quantum algorithms. In this work, we introduce the \textit{Randomized Truncated Series} (RTS) which extends this acceleration to all quantum algorithms that rely on truncated series approximations. RTS offers two key advantages: it quadratically suppresses truncation errors and allows for continuous adjustment of the effective truncation order. By leveraging random mixing between two quantum circuits, RTS ensures that their probabilistic combination accurately realizes the desired algorithm, while significantly reducing the average circuit size. We demonstrate the versatility of RTS through concrete applications. Our results shed light on the path toward practical quantum advantage.
\end{abstract}

\maketitle

\newtheorem{theorem}{Theorem}
\newtheorem{lemma}{Lemma}
\newtheorem{corollary}{Corollary}
\newtheorem{claim}{Claim}
\newtheorem{definition}{Definition}
\newtheorem{conjecture}{Conjecture}
\newtheorem{observation}{Observation}
\newtheorem{presumption}{Presumption}
\newtheorem{proposition}{Proposition}
\newtheorem{remark}{Remark}

\definecolor{BlueSNS}{rgb}{0,0.1098,0.4980}

\newcommand{\U}{F_1}
\newcommand{\UD}{F_1^\dagger}
\newcommand{\UU}{F_2}
\newcommand{\UUD}{F_2^\dagger}
\newcommand{\qi}[1]{\textcolor{blue}{(QI: #1)}}
\makeatletter
\newcommand*{\rom}[1]{\romannumeral #1}
\makeatother
\newcommand{\Id}{\mathds{1}}
\makeatother
\newcommand{\ii}{\mathrm{i}}

Quantum algorithms, such as those for Hamiltonian simulation (HS)~\cite{sethUniveral1996,lowHamiltonianSimulationQubitization2019a,lowQuantumSignalProcessing2017,berrySimulatingHamiltonianDynamics2015,Childs2019fasterquantum,ChildsNearlyOpt2019,ZhaoRandominput2022}, solving differential equations~\cite{berryQuantumAlgorithmLinear2017, liuEfficientQuantumAlgorithm2021, an_quantum_2022, anTheoryQuantumDifferential2023, kroviImprovedQuantumAlgorithms2023, fangTimemarchingBasedQuantum2023}, and singular value transformation~\cite{gilyenQuantumSingularValue2019,sunderhauf2023generalized}, achieve up to exponential asymptotic speedup compared to their classical counterparts~\cite{martynGrandUnificationQuantum2021}. These algorithms serve as the basis for key proposals in chemistry~\cite{ThermalRateChem1999,QchemSimu2001,weckerGatecountEstimatesPerforming2014,babbushEncodingElectronicSpectra2018,babbushLowDepthQuantum2018,YuanReviewChem2020}, materials~\cite{costaQuantumAlgorithmSimulating2019,haahQuantumAlgorithmSimulating2023,mizutaOptimalHamiltonianSimulation2023}, cryptography~\cite{shorPolynomialTimeAlgorithmsPrime1997}, engineering~\cite{liPotentialQuantumAdvantage2023,ameriQuantumAlgorithmLinear2023,lindenQuantumVsClassical2022}, and finance~\cite{hermanSurveyQuantumComputing2022}.

Many algorithms approximately transform an operator $H$ by truncated polynomials (e.g. in HS, matrix inversion, or factoring~\cite{lowQuantumSignalProcessing2017,gilyenQuantumSingularValue2019,martynGrandUnificationQuantum2021}). The truncation order $K$ is derived according to the existing error analysis and a predetermined accuracy. While increasing $K$ enhances precision, it requires additional qubits and gates, resulting in more complex circuits that impede the realization of quantum advantage~\cite{stilck2021limitations, zhouWhatLimitsSimulation2020}. Moreover, the target accuracy generally falls between $\epsilon_K$ and $\epsilon_{K+1}$ achievable by truncation orders $K$ and $K+1$ (see Fig.~\ref{fig:intui_graph}(a)). To satisfy the desired accuracy, $K$ is rounded up to the next integer, leading to more intricate quantum circuits without generally enhancing performance.

Simplifying circuits while preserving precision is crucial and challenging, especially when targeting a broad class of quantum algorithms. We observed that classical randomness accelerates specific quantum algorithms~\cite{Childs2019fasterquantum,campbel2019random,wan2022random,yang2021accelerate,zeng2022simplehighprecisionhamiltoniansimulation,granetHamiltonianDynamicsDigital2024, Hsieh2024double} by replacing faithful execution of the original quantum circuit with randomizing on an ensemble of simpler ones. In this article, we generalize the concept to encompass all quantum algorithms depending on truncated series expansion. Specifically, we introduce \textit{Randomized Truncated Series} (RTS) to achieve a quadratically improved and continuously adjustable truncation error. On the high level, we utilize random mixing of truncated series such that truncation errors cancel out one another. To better approximate $F(H):= \sum_{k=0}^{\infty} \alpha_k H^k$, RTS performs random mixing on $F_1(H):= \sum_{k=0}^{K_1} \alpha_k H^k$ and a modified polynomial of order $K_2 > K_1$, $F_2(H):= \sum_{k=0}^{K_1} \alpha_k H^k + 1/(1-p) \sum_{k=K_1+1}^{K_2} \alpha_k H^k$, where $p \in [0,1)$ is the mixing probability and $\alpha_k$ are real coefficients. One can fine-tune the circuit cost by adjusting the continuous mixing probability, $p$, thereby creating an effective fractional truncated order. Unlike in previous works, a key difference in this setting is the non-unitarity of $F_i$ \footnote{We use the subscript $i\in{1,2}$.}. The mixing lemma proposed in Refs.~\cite{campbellShorterGateSequences2017,hastingsTurningGateSynthesis2016} relies on the invariance of the diamond norm under unitary evolution, thus incompatible with our case. We provide a generalized version to quantify the error for mixing near-unitary operators, which may be of independent interest.

RTS exhibits broad applicability across various quantum algorithms. We demonstrate its utility in optimizing: (\rom{1}) HS via both linear combination of unitaries (LCU)~\cite{berrySimulatingHamiltonianDynamics2015} and (\rom{2}) quantum signal processing (QSP)~\cite{lowHamiltonianSimulationQubitization2019a} frameworks; (\rom{3}) the uniform spectral amplification (USA) algorithm within the QSP framework to showcase the application of RTS to polynomial composition~\cite{lowQuantumSignalProcessing2017}, and (\rom{4}) the solution of differential equations involving truncated series as subroutines~\cite{berryQuantumAlgorithmLinear2017}. Our findings indicate that RTS significantly reduces errors by several orders of magnitude. For instance, when targeting an accuracy of $10^{-8}$, RTS can reduce the error upper bound by four orders of magnitude in accelerating the BCCKS algorithm~\cite{berryQuantumAlgorithmLinear2017}. This substantial error reduction is crucial in high-accuracy regimes required by applications such as simulating chemical reactions~\cite{reiherElucidatingReactionMechanisms2017}, where quantum advantage is anticipated in the near term.

This paper is organized as follows. We present our main result RTS in section~\ref{sec:result}, which contains a protocol to execute RTS and the error bound of performing random mixing. The following section~\ref{sec:app} contains four applications of RTS framework including in LCU, QSP and solving ODE. We give rigorous performance guarantee for these applications, and for the LCU based HS algorithm, we made some modification to improve the success probability tailored to the RTS framework. We also demonstrate the numerical performance of RTS framework on the BCCKS algorithm. We finally conclude our work in section~\ref{sec:sum}. In section~\ref{sec:method}, we present the technical details for the proof in four applications and we defer the proof of some corollaries in section~\ref{sec:method} to appendix~\ref{app:LCU} and \ref{app:QSP}.

\section{Results}
\label{sec:result}

RTS involves mixing two operators with specified probabilities, as shown in Fig.~\ref{fig:intui_graph}(b). We use quantum circuit $\mathfrak{V}_{i}$  to probabilistically implement the near-unitary operator $V_i$, which encodes information of $F_i$. The circuit $\mathfrak{V}_1$ is identical to the standard technique truncated at order $K_1$. The circuit $\mathfrak{V}_2$ differs from a standard truncation at order 
$K_2$ only by a constant amplification of the higher-order terms. Therefore, the implementation generally differs only in the angles of the rotation gates that encode the coefficients. Normally, we need resources scale linearly with $K_{i}$ to implement $\mathfrak{V}_{i}$~\cite{berrySimulatingHamiltonianDynamics2015, childsFirstQuantumSimulation2018}. Thus, the average resource cost scales linearly with $K_{\text{m}} := pK_1 + (1-p)K_2$ with RTS. While $K_{\text{m}}$ may correspond to a higher average truncation order, RTS results in a much lower error and reduces effective truncation order for target accuracy.
\begin{figure}
    \centering
    \includegraphics[width=0.5\linewidth]{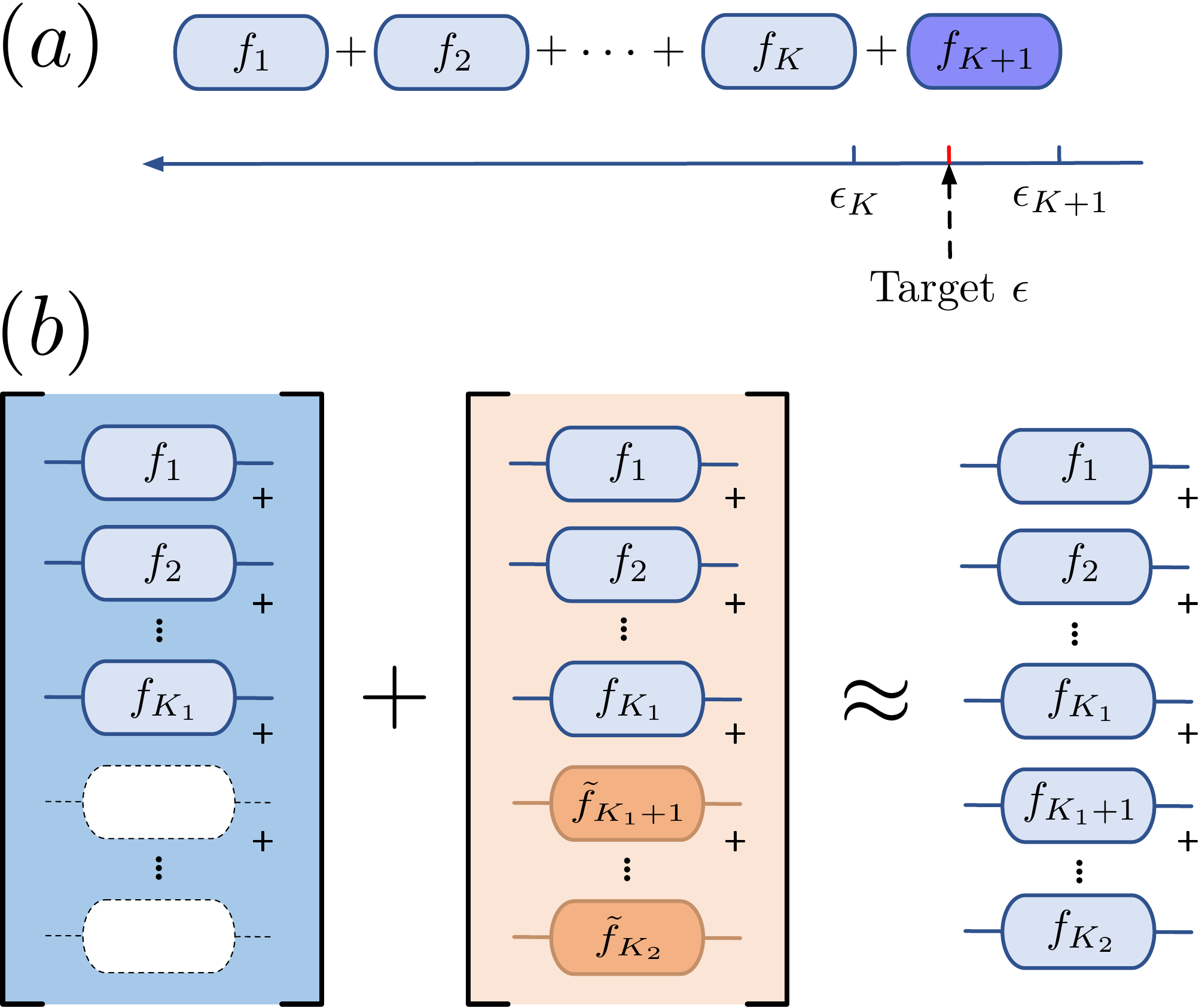}
    \caption{Illustration of conceptual idea. We denote $f_k$ as the $k$-th term in $F(H)$. (a) depicts the conventional approach. The target error $\epsilon$ falls between truncation errors for two series of order $K$ and $K+1$, resulting in inefficiency in the $(K+1)$-th term. (b) demonstrates the RTS method, where we mix two series expansions $F_1$ (depicted in blue brackets) and $F_2$ (depicted in orange brackets) with probability $p$ and $1-p$, respectively. $\tilde{f}_k = 1/(1-p) f_k$ in the orange bracket are modified terms that return to $f_k$ after sampling the measurement result. The amplification coefficient $1/(1-p)$ will be reverted to unity because its contribution to the final result is suppressed by its mixing probability $1-p$. Consequently, the output of RTS includes information on higher-order terms in the series expansion, and better approximates $F(H)$.}
    \label{fig:intui_graph}
\end{figure}

The following protocol outlines procedures to execute RTS. This approach yields only classical measurement outcomes, but we show that one can still retrieve a coherent quantum state when $\mathfrak{V}_{i}$ is structured as concatenated identical segments, referred to as segmented algorithms. For example, Hamiltonian simulation (HS) usually breaks the total evolution $e^{-iHt}$ into $r$ smaller segments $U = e^{-iHt/r}$, so it naturally fits the segmented framework. For segmented algorithms, we handle step 2 differently. Since $\mathfrak{V}_i$ and $U$ evolve the system for the same time duration, we maintain the same number of segments as the original algorithm. Each segment is replaced with either $\mathfrak{V}_1$ or $\mathfrak{V}_2$, where $\mathfrak{V}_1$ is simpler than $U$ while targeting same accuracy. As $\mathfrak{V}_1$ is used for the majority of segments (with probability $p$), the overall circuit depth is reduced.
\begin{enumerate}[leftmargin=*]
    \item \textbf{Random Circuit Generation:} Randomly generate $\mathfrak{V}_1$ and $\mathfrak{V}_2$, with probabilities $p$ and $(1-p)$, respectively. 
    \item \textbf{Circuit Construction:} Apply the prepared quantum circuit to the input state. For segmented algorithms, replace each segment $U$ with a circuit $\mathfrak{V}_i$ generated in step 1. 
    \item \textbf{Post-Selection:} For the succeeded circuits, we execute the follow-up operations, i.e. measure observables. Then statistically combine the classical results.
    \item \textbf{Repeats:} Repeat until the desired sampling accuracy is reached.
\end{enumerate}

Since we approximate a unitary operator by truncated series, the resultant operator need not be a valid unitary operator. We thus define $\epsilon$-near unitary operator to be 
\begin{definition}
An operator $V$ is $\epsilon$‑near‑unitary if there exists a unitary $U$ with $\|V-U\|\le \epsilon$, for $0\le \epsilon\le 1$.
\end{definition}

The original proof in ref.\cite{campbellShorterGateSequences2017, hastingsTurningGateSynthesis2016} applied the unitary invariance of the diamond norm. In our case $V$ need not be unitary, so we perform a renormalization of the output state and obtain a generalized error bound. We generalize the mixing lemma~\cite{Childs2019fasterquantum,campbellShorterGateSequences2017,hastingsTurningGateSynthesis2016}. Define a mixing channel $\mathcal{V}_{\text{mix}}(\rho) = p \mathcal{V}_1(\rho) + (1-p) \mathcal{V}_2(\rho)$ for a density matrix $\rho$, where $\mathcal{V}_i$ is quantum channel corresponds to $V_i$, i.e. $\mathcal{V}_{i}(\rho) = V_{i} \rho V_{i}^\dagger$.

\begin{lemma}
\label{le:mix}
Let $V_1$ and $V_2$ be near-unitary operators approximating an ideal operator $U$. Denote the operator $V_{m} := pV_1 + (1-p)V_2$. Assume the operator norm follows $\|V_1-U \|\le a_1$, $\|V_2-U \|\le a_2$, and $\|V_{m}-U\|\le b$, then the density operator $\rho = \ket{\psi}\bra{\psi}$ acted on by the mixed channel $\mathcal{V}_{\text{mix}}$ satisfies  
\begin{equation}
\left\|\mathcal{V}_{\text{mix}}(\rho) -\mathcal{U}(\rho) \right\|_1  \le \varepsilon, 
\end{equation}
where $\varepsilon= 4b+2pa_1^2+2(1-p)a_2^2$, $\mathcal{U}(\rho) = U\rho U^{\dagger}$ and $\|\cdot\|_1$ is the 1-norm.
\end{lemma}

\begin{proof}
From the assumption of operator norms, we have 
\begin{equation}
    \|V_1\ket{\psi}-U\ket{\psi} \|\le a_1, \quad \|V_2\ket{\psi}-U\ket{\psi} \|\le a_2, \quad \|V_{m}\ket{\psi}-U\ket{\psi}\|\le b
\end{equation}

We denote the non-normalized state $\mathcal{V}_{\text{mix}}(\rho)  = p V_1\ket{\psi}\bra{\psi}V_1^{\dagger} +  (1-p) V_2\ket{\psi}\bra{\psi}V_2^{\dagger}$, $\ket{\epsilon_1}=V_1\ket{\psi}-U\ket{\psi}$, $\ket{\epsilon_2}=V_2\ket{\psi}-U\ket{\psi}$, and $\ket{\epsilon_m}=(pV_1+(1-p)V_2)\ket{\psi}-U\ket{\psi}=p\ket{\epsilon_1}+(1-p)\ket{\epsilon_2}$. 
\begin{equation}
\begin{aligned}
\mathcal{V}_{\text{mix}}(\rho) -U\ket{\psi}\bra{\psi}U^{\dagger}& = p  (U\ket{\psi}+\ket{\epsilon_1})(\bra{\epsilon_1}+\bra{\psi}U^{\dagger} )+ (1-p)  (U\ket{\psi}+\ket{\epsilon_2})(\bra{\epsilon_2}+\bra{\psi}U^{\dagger} )- U\ket{\psi}\bra{\psi}U^{\dagger} \\
&=\ket{\epsilon_m}\bra{\psi}U^{\dagger}+ U\ket{\psi} \bra{\epsilon_m}+ p\ket{\epsilon_1}\bra{\epsilon_1}+ (1-p)\ket{\epsilon_2}\bra{\epsilon_2}.
\end{aligned}
\end{equation}
According to the definitions, $\|\ket{\epsilon_1}\|\le a_1$, $\|\ket{\epsilon_2}\|\le a_2$, $\|\ket{\epsilon_m}\|=\sqrt{\braket{\epsilon_m|\epsilon_m}}\le b$, 
\begin{equation}
\begin{aligned}
\|\mathcal{V}_{\text{mix}}(\rho) -U\ket{\psi}\bra{\psi}U^{\dagger}\|_1
&\le \|\ket{\epsilon_m}\bra{\psi}U^{\dagger}\|_1 + \|U\ket{\psi} \bra{\epsilon_m}\|_1+ p \|\ket{\epsilon_1}\bra{\epsilon_1}\|_1+ (1-p)\|\ket{\epsilon_2}\bra{\epsilon_2}\|_1\\
&\le 2\sqrt{\braket{\epsilon_m|\epsilon_m}}+p \braket{\epsilon_1|\epsilon_1} + (1-p) \braket{\epsilon_2|\epsilon_2}\\
&\le 2b+pa_1^2+(1-p)a_2^2 =: \varepsilon^\prime .
\end{aligned}
\end{equation}
With $|\varepsilon^\prime| \le 1$, this also implies   $1-\varepsilon^\prime \le \|\mathcal{V}_{\text{mix}}(\rho)\|_1\le 1+\varepsilon^\prime$
\begin{equation}
\left\|\frac{\mathcal{V}_{\text{mix}}(\rho)}{\|\mathcal{V}_{\text{mix}}(\rho)\|_1}-U\ket{\psi}\bra{\psi}U^{\dagger} \right\|_1   \le \left\|\mathcal{V}_{\text{mix}}(\rho) -U\ket{\psi}\bra{\psi}U^{\dagger}\right\|_1+ \|\mathcal{V}_{\text{mix}}(\rho)\|_1(\frac{1}{\|\mathcal{V}_{\text{mix}}(\rho)\|_1}-1)
\le 2\varepsilon^\prime.
\end{equation}

Therefore, with $\tilde{\mathcal{V}}_{mix}(\rho) = \mathcal{V}_{\text{mix}}(\rho)/{\rm Tr}\mathcal{V}_{\text{mix}}(\rho)$ is the normalised quantum state and $\mathcal{U}(\rho) = U \rho U^\dagger$, we have 
\begin{equation}
\label{eq: mixing_channel_1_norm}
    \left\|\tilde{\mathcal{V}}_{mix}(\rho) - \mathcal{U}(\rho)\right\|_1 \le 2\epsilon^\prime =: \varepsilon.
\end{equation}

\end{proof}

We utilize the Lemma~\ref{le:mix} and the structure of $V_i$ to analyze the performance of RTS and obtain the main Theorem.

\begin{theorem}
    \label{cor:main}
    Let $U = \sum_{k=0}^{\infty} \alpha_k H^k$ be an operator in series expansion form. Assume a quantum circuit $\mathfrak{V}_1$ encodes the truncated operator $V_1$ such that $\|U - V_1\| \le a_1$, and there exist another quantum circuit $\mathfrak{V}_2$ that encodes $V_2$, where $\|U - V_2\| \le a_2$ and $a_2 = \mathcal{O}\left(a_1\right)$. Employing RTS on $V_1$ and $V_2$ yields a mixing channel $\mathcal{V}_{\text{mix}}$ such that 
    \begin{equation}
       \left\|\mathcal{V}_{\text{mix}}(\rho)-\mathcal{U}(\rho)\right\|_1 = \mathcal{O}\left(a_1^2\right).
    \end{equation}
\end{theorem}

\begin{proof}
We have assumed $V_1$ and $V_2$ have the form
\begin{equation}
    V_1 = \sum_{k=0}^{K_1} \alpha_k H^k, \quad  V_2 = \sum_{k=0}^{K_1} \alpha_k H^k + \frac{1}{1-p}\sum_{k=K_1+1}^{K_2} \alpha_k H^k,
\end{equation}
for some $K_2 \ge K_1,$ and $ K_1,K_2 \in \mathbb{N}$. 
Therefore, we have 
\begin{equation}
    V_{\rm m} = pV_1 + (1-p)V_2 = \sum_{k=0}^{K_2} \alpha_k H^k
\end{equation}
We can then calculate the error for all operators
\begin{equation}
\begin{split}
    b = \left\|U - (pV_1 + (1-p)V_2)\right\| &= \left\|\sum_{k=0}^{\infty} \alpha_k H^k -  p\sum_{k=0}^{K_1} \alpha_k H^k - (1-p)\left(\sum_{k=0}^{K_1} \alpha_k H^k + \frac{1}{1-p}\sum_{k=K_1+1}^{K_2} \alpha_k H^k\right)\right\| \\
    &= \left\|\sum_{k=0}^{\infty} \alpha_k H^k - \sum_{k=0}^{K_2} \alpha_k H^k\right\| = \sum_{k=K_2+1}^{\infty} \alpha_k \|H\|^k\\
    a_1 = \left\|U-V_1\right\| &=  \left\|\sum_{k=0}^{\infty} \alpha_k H^k -  \sum_{k=0}^{K_1} \alpha_k H^k\right\| =  \sum_{k=K_1+1}^{\infty} \alpha_k \|H\|^k \\
    a_2 = \left\|U-V_2\right\| &= \left\|\sum_{k=0}^{\infty} \alpha_k H^k - \sum_{k=0}^{K_1} \alpha_k H^k - \frac{1}{1-p}\sum_{k=K_1+1}^{K_2} \alpha_k H^k\right\| \\
    &= \left\|\sum_{k=0}^{\infty} \alpha_k H^k - \sum_{k=0}^{K_1} \alpha_k H^k +\sum_{k=K_1+1}^{K_2} \alpha_k H^k - \sum_{k=K_1+1}^{K_2} \alpha_k H^k - \frac{1}{1-p}\sum_{k=K_1+1}^{K_2} \alpha_k H^k\right\| \\
    &=  \left\|\sum_{k=0}^{\infty} \alpha_k H^k - \sum_{k=0}^{K_2} \alpha_k H^k\right\|  + \left\|\left(\frac{1}{1-p} - 1\right)\sum_{k=K_1+1}^{K_2} \alpha_k H^k\right\| \\
    &\le b + \frac{p}{1-p}a_1
\end{split}
\end{equation}

Applying lemma \ref{le:mix}, we obtain the error upper-bound, $\epsilon$ of an algorithm after the mixing channel being 
\begin{equation}
\begin{split}
    \epsilon &= 4b + pa_1^2 + (1-p)a_2^2 \\
    &= 4b + pa_1^2 + (1-p)\left(b + \frac{p}{1-p}a_1\right)^2 \\
    &= \mathcal{O}\left(a_1^2\right)
\end{split}
\end{equation}

In the last line, we assume that the truncation error is reduced exponentially with the truncation order in most series expansions. Therefore, $b$ can be neglected in the big O notation. 
\end{proof}

In Theorem \ref{cor:main}, we neglect $b = \|U-V_{\rm m}\|$ in the asymptotic regime since higher-order terms beyond $K_2$ are exponentially suppressed for typical series expansions. The cost reduction delivered by the Randomised Truncated Series (RTS)
framework manifests differently across quantum algorithms. We will show that the cost reduction for the BCCKS algorithm and the QSP-based HS is reduced by 50\%, and the cost reduction for other algorithms, where cost depends differently with truncation order can be analyzed accordingly. Specifically, for algorithms with cost polynomially depend on truncation order, the cost reduction can be polynomial.

The $K$-truncated BCCKS algorithm for simulating a
$d$-sparse, $n$-qubit Hamiltonian~$H$ for an evolution time~$t$ incurs a
two–qubit–gate cost of $\mathcal{O}  \bigl(n\,\tau\,
\log(\tau/\epsilon)\,K\bigr)$, where $\tau = d^{2}\|H\|_{\max}t$.  In
addition it queries the Hamiltonian oracle $\mathcal{O}(\tau K)$ times.
Crucially, both costs scale linearly with the truncation order
$K  =  \log(\tau/\epsilon)/\log  \bigl(\log(\tau/\epsilon)\bigr)$.

Let the target error be $0<\varepsilon\ll1$ and the fixed problem
constant be $\tau  \gg  1$.  We define
\begin{equation}\label{eq:Cdef}
  K(\varepsilon;\tau)
  \;=\;
  \frac{\log(\tau/\varepsilon)}{\log  \bigl(\log(\tau/\varepsilon)\bigr)}
  \;\;.
\end{equation}

By definition, achieving accuracy~$\varepsilon$ with the original
BCCKS algorithm requires a cost
$K_{\mathrm{orig}}(\varepsilon;\tau)  =  K(\varepsilon;\tau)$.
RTS, by design, achieves quadratic error suppression, thus we run the
same algorithm to precision~$\sqrt{\varepsilon}$ and then apply
statistical mixing.  Its cost is therefore
\begin{align*}
    K_{\mathrm{mix}}(\varepsilon;\tau)
    &= K(\sqrt{\varepsilon};\tau)
      = \frac{\log(\tau/\sqrt{\varepsilon})}
             {\log  \bigl(\log(\tau/\varepsilon)\bigr)} \; .
\end{align*}

Writing $A  =  \log\tau$ and $L  =  \log(1/\varepsilon)$, with
$L  \to  \infty$ as $\varepsilon  \to  0$, we have
\begin{align*}
  K_{\mathrm{orig}} &= \frac{A+L}{\log(A+L)}, &
  K_{\mathrm{mix}} &= \frac{A+L/2}{\log(A+L/2)}\; .
\end{align*}
Hence
\begin{equation*}
  \frac{K_{\mathrm{mix}}}{K_{\mathrm{orig}}}
  = \frac{A+\tfrac12 L}{A+L}
    \;\frac{\log(A+L)}{\log(A+\tfrac12 L)}
  = \frac12\,\frac{\log L}{\log(L/2+A)}\; .
\end{equation*}

Because $A$ is constant while $L  \to  \infty$, we have
$\log(L/2+A)\sim\log L$, so the fraction tends to~$1$; consequently
\begin{equation*}
  \frac{K_{\mathrm{mix}}}{K_{\mathrm{orig}}}
  \xrightarrow[\varepsilon\to0]{} \frac12\; .
\end{equation*}
In the asymptotic regime the RTS framework therefore 
halves the leading‑order cost of the BCCKS algorithm.  The same
line of reasoning applies to the truncated‑series ODE solver of
Berry~\emph{et~al.}, whose resource requirements also scale linearly
with the truncation order.

For QSP based HS the result is
similar but the derivation differs slightly.  Following the cost model
of Ref.~\cite{lowHamiltonianSimulationQubitization2019a} (noting that
Ref.~\cite{gilyenQuantumSingularValue2019} identifies a minor error
that does not affect the leading behaviour), the truncation error obeys
\[\epsilon \ge \bigl(4|t|^{K}\bigr)/(K!\,2^{K})=\mathcal{O}  \bigl((et/K)^{K}\bigr),\]
where $t$ is the simulation time and $K$ is the Jacobi–Anger order.
Thus $\log(1/\epsilon)=\mathcal{O}\bigl(K\,\log(2K/(et))\bigr)$ and one
can bound $K=\mathcal{O}\bigl(t+\log(1/\epsilon)\bigr)$.  With RTS we
replace $\epsilon$ by $\sqrt{\epsilon}$, giving
$\log(1/\sqrt{\epsilon}) = \tfrac12\log(1/\epsilon)$ and therefore
$K\to K/2$ for fixed accuracy; the associated cost is therefore
halved. We next demonstrate how to utilize RTS and the performances with several examples. In each instance, we may redefine variables to avoid using lengthy subscripts.

\section{Applications}
\subsection{BCCKS example} 
\label{sec:app}
Hamiltonian simulation~\cite{berrySimulatingHamiltonianDynamics2015} (HS) is one of the fundamental quantum algorithms. Moreover, it acts as subroutines in algorithms like quantum phase estimation, quantum linear system solver, etc. Therefore, accurate and efficient HS is crucial in both near- and long-term perspectives. 

The Taylor expansion of unitary evolution under the system Hamiltonian $H$ for time $t$ can be written as
\begin{equation}
    \label{eq: LCU_u}
    U = e^{-\ii Ht} =\sum^{\infty}_{k=0}\frac{\left(-\ii Ht\right)^k}{k!}.
\end{equation}
In this scenario, we have $F_1 := \sum^{K_1}_{k=0}\frac{\left(-\ii Ht\right)^k}{k!}$, and $F_2 := \sum^{K_1}_{k=0}\frac{\left(-\ii Ht\right)^k}{k!} + \frac{1}{1-p}\sum^{K_2}_{k=K_1+1}\frac{\left(-\ii Ht\right)^k}{k!}$. The BCCKS algorithm is a typical segmented algorithm that aims to implement $F_1$. Assume an $n$-qubit Hamiltonian $H$ can be decomposed into a sum of efficiently simulatable unitaries $H_l$ with coefficients $\alpha_l$, i.e. $H = \sum^L_{l=1}\alpha_l H_l$.  Then, we can re-express $F_{i}$ in the form $F_{i}=\sum_{j=0}^{\Gamma-1} \beta^{(i)}_j \tilde{V}^{(i)}_j$,
where $\Gamma = \sum_{k=0}^{K_{i}} L^k$, $\tilde{V}^{(i)}_j$ represents one of the unitaries $(- \ii )^kH_{l_1}\cdots H_{l_k}$, and $\beta^{(i)}_j$ is the corresponding positive coefficient. $F_{i}$ is in standard LCU form~\cite{KothariRobin2014} which can be implemented by \textbf{SELECT} and \textbf{PREPARE} oracles followed by oblivious amplitude amplification (OAA)~\cite{berryExponentialImprovementPrecision2014}. For implementation details, one can refer to~\cite{berrySimulatingHamiltonianDynamics2015,childsFirstQuantumSimulation2018}. 
 
Applying oblivious amplitude amplification to $F_i$ yields modified circuits $V_i$ that apply the approximate transformation with high probability. Due to the additional term in $F_2$, we have to apply OAA with an additional round of reflection, and the resulting quantum circuits perform the following transform 
\begin{equation}
\label{eq:LCU_circuit}
    \begin{aligned}
        \ket{0}\ket{\psi} &\mapsto \ket{0}V_{i}\ket{\psi} + \ket{\perp_{i}},\\
        V_1 &:= \frac{3}{s_1}F_1-\frac{4}{s^3_1}F_1F_1^\dagger F_1, \\
        V_2 &:= \frac{5}{s_2}F_2-\frac{20}{s_2^3}F_2F_2^\dagger F_2 + \frac{16}{s_2^5}F_2F_2^\dagger F_2F_2^\dagger F_2,
    \end{aligned}
\end{equation}
where $(\mathds{1}\otimes\bra{0})\ket{\perp_{i}} = 0$, $s_1 \approx 2$ and $s_2 \approx (\sin(\pi/10))^{-1}$. Thus, projecting on $\ket{0}$ in the first register by post-selection essentially implements $V_{i}$. The error bound, cost, and failure probability for RTS implementing the BCCKS algorithm are given by the following corollary. 

\begin{corollary}
\label{cor:LCU}
    let $V_1$ and $V_2$ be as defined in Eq.~\eqref{eq:LCU_circuit}. Then for any mixing probability $p\in[0,1)$ and input state $\rho$, the output of the mixed channel $\mathcal{V}_{\text{mix}}(\rho) = p V_1\rho V_1^{\dagger} +  (1-p) V_2\rho V_2^{\dagger}$ and an ideal evolution for a segment, $U = e^{-iH\tau}$, is bounded by 
    \begin{equation}
    \left\|\mathcal{V}_{\text{mix}}(\rho) - U\rho U^\dagger\right\|_1 \le \max\left\{\frac{40}{1-p}\delta_1^2, 8\delta_m\right\},
    \end{equation}
    where $\delta_1 = 2\frac{(\ln 2)^{K_1+1}}{(K_1+1)!}$ and $\delta_m = 2\frac{(\ln 2)^{K_2+1}}{(K_2+1)!}$. The overall cost of implementing this segment is 
    \begin{equation}
        G=\tilde{\mathcal{O}}( n L (p K_1+ (1-p)K_2)),
    \end{equation}
    where $n$ is number of qubit, $L$ is the number of terms in the unitary expansion of $H$, $K_1$ and $K_2$ are truncation order in $V_1$ and $V_2$ respectively, and $\tilde{\mathcal{O}}$ means polylog factors are suppressed. The failure probability corresponds to one segment being upper bounded by $\xi \le \frac{8}{1-p}\delta^2_1 + 4\delta_1$.
\end{corollary}

For the BCCKS algorithm, which employs unary encoding in the \textbf{SELECT} circuit, the ancillary-qubit width is $\mathcal{O}(K\log L)$.  In our RTS protocol, $\mathfrak{V}_1$ uses a smaller width, whereas $\mathfrak{V}_2$ requires a larger ancillary system.  In each segment, the ancilla qubit is measured then refreshed with a clean qubit. Therefore, the total spatial cost can be evaluated using the average ancilla width per segment. The total number of qubits that must be prepared during executing the algorithm equals the segment count multiplied by that average width.  Since the number of segments is unchanged but the average width is reduced, the burden on the ancillary system is ultimately lower.

\subsection{Quantum Signal Processing (QSP) examples}

QSP is powerful in transforming the eigenvalue of a Hamiltonian $H$. We will demonstrate the application of RTS in two settings: exponential-function transformation and truncated-linear transformation. Other algorithms relying on QSP can be addressed similarly.

Consider a single eigenstate $\ket{\lambda}$ of $H$. QSP performs a degree-$d$ polynomial transformation $f(\lambda)$ by classically finding a vector of angles $\vec{\phi} \in \mathbb{R}^d$ and construct an iterator $W_{\vec{\phi}}$ such that 

\begin{equation}
    \begin{aligned}
        W_{\vec{\phi}} = \begin{pmatrix}
            f(\lambda) & \cdot \\
            \cdot & \cdot
        \end{pmatrix},
    \end{aligned}
\end{equation}
where $\cdot$ denotes the other matrix elements (discarded through post-selection). The notation is consistent with Ref.~\cite{lowHamiltonianSimulationQubitization2019a}. Regarding HS, we use $f(\lambda)$ to approximate $U_\lambda := e^{-\ii\lambda t}$ by the truncated Jacobi-Anger expansion~\cite{abramowitz1988handbook}. Utilizing RTS, we obtained the following corollary. 

\begin{corollary}
    \label{cor:QSP_HS}
    Consider two quantum circuits implementing $\hat{V}_1$ and $\hat{V}_2$ in Eq.~\eqref{eq:formal_qsp_HS}.  Given a mixing probability $p \in [0,1)$ and an arbitrary density matrix $\rho$, distance between the evolved state under the mixing channel $\mathcal{V}_{\text{mix}}(\rho) = p V_1\rho V_1^{\dagger} +  (1-p) V_2\rho V_2^{\dagger}$ and an ideal evolution for $U_\lambda$ is bounded by 
    \begin{equation}
    \left\|\mathcal{V}_{\text{mix}}(\rho) - U\rho U^\dagger\right\| \le \max\left\{28\delta_1, 8\sqrt{\delta_m}\right\},
    \end{equation}
    where $\delta_m = \frac{4t^{K_2}}{2^{K_2} K_2!}$ and $\delta_1 = \frac{4t^{K_1}}{2^{K_1} K_1!}$. The overall cost is 
    \begin{equation}
        G=\mathcal{O}\left(pK_1 + (1-p)K_2\right),
    \end{equation}
    where $d$ is the sparsity of $H$, and $K_1$ and $K_2$ are truncated order in $V_1$ and $V_2$ respectively. The failure probability is upper bounded by $\xi \ge 4p \sqrt{\delta_2}$.
\end{corollary}

The other algorithm we demonstrate under the context of QSP is the uniform spectral amplification (USA)~\cite{lowQuantumSignalProcessing2017}, which is a generalization of amplitude amplifications~\cite{nagaj2009fast,berryExponentialImprovementPrecision2014} and spectral gap amplification~\cite{SommaSpectral2013}. This algorithm amplifies the eigenvalue by $1/{2\Gamma}$ if $|\lambda| \in [0,\Gamma]$ while maintaining $H$ normalized. Specifically, USA approximates the truncated linear function 
\begin{equation}
\label{eq:QSP_linear_function}
    f_{\Gamma}(\lambda) = \begin{cases}
        \frac{\lambda}{2\Gamma}, & |\lambda| \in [0,\Gamma] \\
        \in [-1,1], &  |\lambda| \in (\Gamma,1],
    \end{cases}
\end{equation}
Eq.~\eqref{eq:QSP_linear_function} is approximated by $\tilde{f}_{\Gamma,\delta}(\lambda)$, where $\delta = \displaystyle\max_{|x| \in [0,\Gamma]} \left|\tilde{f}_\Gamma(x)  - x/(2\Gamma)\right|$ is the maximum error tolerance, formed by composing a truncated Jacobi–Anger expansion approximation of the error functions. Using RTS, error and circuit cost are given by the following corollary.

\begin{corollary}
    \label{cor:USA}
    For the unitary $\hat{W}_{\vec{\varphi}_{1(2)}}$ and post-selection scheme that implementing the operator in Eq.~\eqref{ep:QSP_lin_v_1(2)} (see Methods), there exist two sets of angles $\vec{\varphi}_{1(2)}$ such that $D_{1(2)}(\lambda) = \hat{P}_{\Gamma,\delta,K_{1(2)}}(\lambda)$. Denote these two quantum circuits as $V_1$ and $V_2$ respectively. Then, given a mixing probability $p \in [0,1)$ and an arbitrary density matrix $\rho$, distance between the evolved state under the mixing channel $\mathcal{V}_{\text{mix}}(\rho) = p V_1\rho V_1^{\dagger} +  (1-p) V_2\rho V_2^{\dagger}$ and an ideal transformation implementing the transformation given by Eq.~\eqref{eq:QSP_TLF} is bounded by 
    \begin{equation}
    \left\|\mathcal{V}_{\text{mix}}(\rho) - f_{\Gamma,\delta}(\rho)\right\| \le \max\left\{8\delta_m,\frac{4}{1-p}\delta_1^2\right\},
    \end{equation}
    where $f_{\Gamma,\delta}(\rho) = f_{\Gamma,\delta}(H)\rho f_{\Gamma,\delta}(H)^\dagger$ is the ideal quantum channel, $\delta_1 =  \frac{8\Gamma e^{-8\Gamma ^2}}{\sqrt{\pi}} \frac{4(8\Gamma^2)^{K_1/2}}{2^{K_1/2}(K_1/2)!}$ and $\delta_m = \frac{8\Gamma e^{-8\Gamma ^2}}{\sqrt{\pi}} \frac{4(8\Gamma^2)^{K_2/2}}{2^{K_2/2}(K_2/2)!}$. The overall cost is 
    \begin{equation}
        G=\mathcal{O}\left(pK_1 + (1-p)K_2\right),
    \end{equation}
    where $d$ is the sparsity of $H$, and $K_1$ and $K_2$ are truncated order in $V_1$ and $V_2$ respectively.
\end{corollary}

Since the ancillary qubit required to perform QSP does not depend on the truncation order $K$, the number of ancillary qubits is not affected by employing RTS.

\subsection{Ordinary Differential Equations example}

Solving differential equations~\cite{kroviImprovedQuantumAlgorithms2023,berryQuantumAlgorithmLinear2017,berryHighorderQuantumAlgorithm2014,liuDenseOutputsQuantum2023} is another promising application of quantum computing, empowering numerous applications.

Consider an anti-Hermitian operator $A$ and the differential equation of the form $d\vec{x}/dt = A\vec{x}+ \vec{b}$, where $A \in \mathbb{R}^{n\times n}$ and $\vec{b} \in \mathbb{R}^n$ are time-independent. The exact solution is given by
\begin{equation}
\label{eq:ODE}
    \vec{x}(t) = e^{At}\vec{x}(0) + \left(e^{At}-\mathds{1}\right)A^{-1}\vec{b},
\end{equation}
where $\mathds{1}$ is the identity matrix. 

We can approximate $e^{z}$ and $\left(e^{z}-\mathds{1}\right)z^{-1}$ by two $K_1$-truncated Taylor expansions:
\begin{equation}
\label{eq:ODE_functions}
\begin{aligned}
    T_{K_1}(z) &:= \sum^{K_1}_{k=0} \frac{z^k}{k!} \approx e^{z},\\
    S_{K_1}(z) &:= \sum^{K_1}_{k=1}\frac{z^{k-1}}{k!} \approx \left(e^{z}-1\right)z^{-1}.
\end{aligned}
\end{equation}

Given non-negative integer $j$, denote $\vec{x}^j$ as the approximated solution at time $jh$ for a short time step $h$ with $\vec{x}^0 = \vec{x}(0)$. We can calculate $\vec{x}^j$ by the recursive relation 
\begin{equation}
\label{eq:ODE_recur}
    \vec{x}_1^{j} = T_{K_1}(Ah)\vec{x}^{j-1} + S_{K_1}(Ah)h\vec{b}.
\end{equation}
Furthermore, we encode the series of recursive equations in a large linear system $\mathcal{L}_1$ as proposed in Ref.~\cite{berryQuantumAlgorithmLinear2017} and denote the operator solving the linear system as $V_1$. Sampling results on $V_1$ give information of $\vec{x}^{j}$ for time step $j$. To employ RTS, we construct another circuit $V_2$ that encodes $\vec{x}_2^j = T_{K_2}(Ah)\vec{x}_2^{j-1} + S_{K_2}(Ah)h\vec{b}$, where $T_{K_2}(z)$ and $S_{K_2}(z)$ are modified expansions with maximum order $K_2$ in another linear system $\mathcal{L}_2$. RTS mixes the solution to $\mathcal{L}_1$ and $\mathcal{L}_2$ with probability $p$ and $1-p$ respectively to give $\vec{x}_{\text{mix}}^j$.
\begin{figure*}
    \centering
    \includegraphics[width = .9\linewidth]{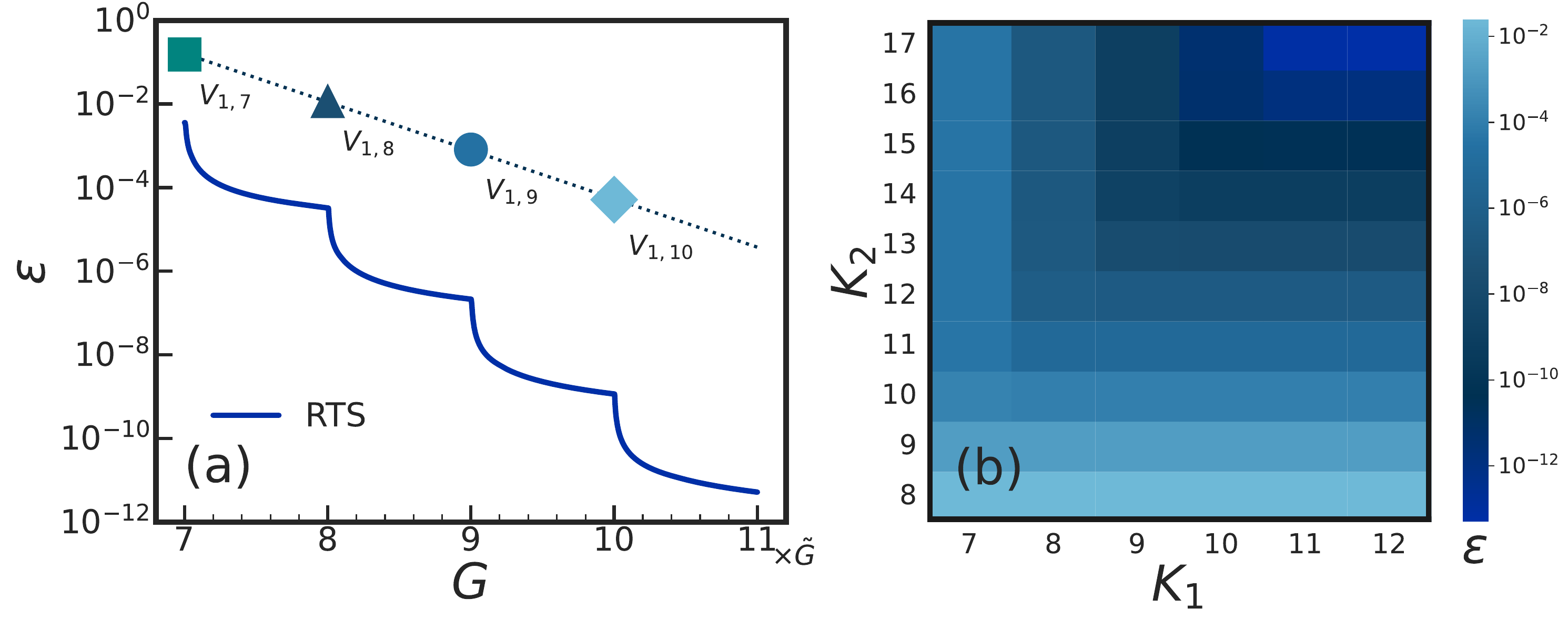}
    \caption{(a) Performance enhancement achieved by applying RTS to the BCCKS algorithm. We denote $V_{1,K1}$ as the performance of the BCCKS algorithm with truncation order $K_1$. With RTS, the overall error exhibits a substantial reduction of several orders of magnitude, consistent with quadratic error suppression. Each point on the curve represents the error obtained using the optimal set of parameters $\{K_1,K_2,p\}$. $\tilde{G} = 131574240$ is a multiplier for CNOT gate cost. (b) illustrates the variation of epsilon with $K_1$ and $K_2$ for a fixed $p = 0.8$.}
    \label{fig:error_cost}
\end{figure*}

\begin{corollary}
\label{cor:APP_ODE}
    Suppose $V_1$ and $V_2$ are quantum circuits solving the linear system in Eq.~\eqref{eq:ODE_LS} with $C_{m,K_1,p}(A)$ and $\tilde{C}_{m,K_1,K_2,p}(A)$ respectively. Solutions at time $jh$ are denoted by $x_1^j$ and $x_2^j$. We apply our framework $\mathcal{V}_{\text{mix}}(\rho) = p V_1\rho V_1^{\dagger} +  (1-p) V_2\rho V_2^{\dagger}$ for a mixing probability $p \in [0,1)$, and denote the obtained solution as $\vec{x}_{\text{mix}}^j$. We can upper bound the estimation error by
    \begin{equation}
        \left\|\vec{x}_{\text{mix}}^j-\vec{x}(jh)\right\| \le \max\left\{8\delta_m, \frac{4}{1-p}\delta_1^2\right\},
    \end{equation}
    where $\delta_1 \le \frac{\mathcal{C}_j}{(K_1+1)!}$, and $\delta_m \le \frac{\mathcal{C}_j}{(K_2+1)!}$. $\mathcal{C}_j$. $\mathcal{C}_j$ is a problem specific constant. 
\end{corollary}

Because the size of the linear system grows with the truncation order, solving $\mathcal{L}_2$ can require a larger quantum circuit than before. Nevertheless, as in the BCCKS case, the total number of qubits engaged over the entire execution is actually lower, so the overall spatial resource is reduced.  The trade-off is that the instantaneous circuit width, which is the peak qubit count at one execution, can be higher.

\subsection{Numerical result}
To validate the improvements, we illustrate RTS by employing it in the HS with the BCCKS algorithm~\cite{berrySimulatingHamiltonianDynamics2015} for the following Ising model with $n = 100$ and $t = 100$
\begin{equation}
\label{eq:simulation_H}
    H = \sum_{i=1}^{n} \sigma_i^{x}\sigma_{i+1}^x  + \sum^n_{i=1}\sigma_i^z,
\end{equation}
where $\sigma_i$ are Pauli's operators acting on the $i^{th}$ qubit and we choose to set all interaction and external field parameters to be $1$ for simplicity. We can decompose Eq.~\eqref{eq:simulation_H} into $L = 200$ Pauli operators, and the coefficient of each Pauli is 1. Thus, we must separate the simulation into $r = \sum_i\alpha_i t/\log(2) = 28854$ segments. We neglect that the evolution time for the last segment is less than $t/r$ for simplicity. 

We focus on the dominant gate cost by ${\rm Select}(H)$, for each segment, we need to perform $3 (4)$ ${\rm Select}(H)$ oracles for each of $V_1(V_2)$ defined in Eq.~\eqref{eq:LCU_circuit}. We neglect any extra cost for implementing the Hermitian conjugate. Each ${\rm Select}(H)$ oracle can be implemented by $K(7.5\times2^w + 6w - 26)$ CNOT gates \cite{childsFirstQuantumSimulation2018}, where $K$ is the truncation order and $w = \log_2(L)$. The CNOT gate cost for faithfully implementing the LCU algorithm at truncation $K$ is thus well approximated by $\tilde{G} = 3rK(7.5\times2^w + 6w - 26)$. We define $G = \tilde{G}/(3r(7.5\times2^w + 6w - 26))$ as a cost indicator since it removes all constants. Note that implementing $V_2$ costs $4/3 G$ more than $V_1$.

Given a fixed CNOT gate cost budget $G$, we exhaustively search through all feasible sets $\{K_1,K_2,p\}$ that consume the entire budget and identify the minimum error upper bound $\epsilon$, with $K_1 \in [1,100]$, $K_2\in [K_1+1,100]$ and $p \in [0,1)$ such that $pK_1 + (1-p)(4/3)K_2 = G$. The results are presented in Fig.~\ref{fig:error_cost}(a). RTS achieves the same accuracy with a reduced gate count. For instance, targeting $\epsilon = 10^{-8}$, we achieve a CNOT-gate savings of approximately 30\%. This cost reduction becomes more significant for larger problem sizes or higher accuracy requirements as can be seen from the trends in Fig.~\ref{fig:error_cost}. We also calculate the cost saving for various $\epsilon$ and achieved over 40\% cost saving. The continuous nature of the error-cost curve highlights the flexibility of RTS in finding optimal parameter combinations without excess gate usage, in contrast to the discreteness presented in the BCCKS algorithm. Fig.~\ref{fig:error_cost}(b) delves deeper into the origins of this error reduction. We observe that for a fixed $K_2$, increasing $K_1$, (which incurs a higher cost) leads to a smaller error reduction compared to increasing $K_2$, (where the cost is mitigated by $p$). Consequently, we can exploit this principle to realize the aforementioned quadratic error improvement. 

The foregoing analyses are strictly asymptotic. In practice the
actual saving depends on sub‑leading terms that are suppressed in the
$\mathcal{O}(\cdot)$ notation.  Consequently the saving may not reach 50\% in specific finite examples.  For the numerical case study on the BCCKS algorithm our simulation saturates at approximately 40~\%
(Table~\ref{tab:cost-saving}).

\begin{table}[htbp]
\centering
\begin{tabular}{|c|c|c|c|}
\hline
\textbf{Error} & \textbf{Framework Cost ($G$)} & \textbf{Original Cost ($G$)} & \textbf{Cost Saving (\%)} \\
\hline
$1.0 \times 10^{-4}$  & 7.29 & 10  & 27.1 \\
$1.0 \times 10^{-8}$  & 9.12 & 13  & 29.9 \\
$1.0 \times 10^{-12}$ & 11.02 & 16 & 31.1 \\
$1.0 \times 10^{-16}$ & 12.59 & 19 & 33.8 \\
$1.0 \times 10^{-20}$ & 14.05 & 22 & 36.1 \\
$1.0 \times 10^{-24}$ & 15.73 & 24 & 34.5 \\
$1.0 \times 10^{-28}$ & 17.04 & 27 & 36.9 \\
$1.0 \times 10^{-32}$ & 18.28 & 29 & 37.0 \\
$1.0 \times 10^{-36}$ & 20.04 & 32 & 37.4 \\
$1.0 \times 10^{-40}$ & 21.06 & 34 & 38.1 \\
$1.0 \times 10^{-44}$ & 22.41 & 37 & 39.4 \\
$1.0 \times 10^{-48}$ & 24.02 & 39 & 38.4 \\
$1.0 \times 10^{-52}$ & 25.04 & 41 & 38.9 \\
$1.0 \times 10^{-56}$ & 26.14 & 43 & 39.2 \\
$1.0 \times 10^{-60}$ & 27.96 & 45 & 37.9 \\
$1.0 \times 10^{-64}$ & 29.02 & 48 & 39.5 \\
$1.0 \times 10^{-68}$ & 30.04 & 50 & 39.9 \\
$1.0 \times 10^{-72}$ & 31.07 & 52 & 40.3 \\
$1.0 \times 10^{-76}$ & 32.34 & 54 & 40.1 \\
\hline
\end{tabular}
\caption{Comparison of framework cost and original cost at varying error thresholds.  The dimensionless unit $\tilde{G}=131\,574\,240$ can be interpreted as an effective truncation order.}
\label{tab:cost-saving}
\end{table}
 
\section{Discussion}
\label{sec:sum}
We have presented the RTS framework, which can accelerate a broad range of quantum algorithms that rely on truncated series approximations. RTS enables a ``fractional'' truncation order and provides a quadratic improvement on $\epsilon$. Essentially, we developed a random mixing protocol with two input quantum circuits $V_1$ and $V_2$. Their truncation errors cancel out one another during the mixing channel, and with the newly introduced mixing probability $p$, continuous adjustment of the overall cost becomes viable. We specifically exhibit the implementation of RTS in the context of HS, uniform spectral amplification, and solving time-independent ODE. These illustrate the flexibility of RTS to be embedded in other algorithms as presented in ODE implementation, and to encompass subroutines like oblivious amplitude amplification and the linear combination of multiple polynomials. Finally, we evaluated the reduction of the CNOT gate in the BCCKS algorithm.

The acceleration achieved by RTS stems from the introduction of classical randomness into the quantum circuit. By leveraging prior classical knowledge of operators like knowing the high-order terms of the series, we can randomize them and save quantum resources. It is important to note that RTS is compatible with other circuit optimization protocols, potentially leading to further speedups like quantum-circuit synthesis~\cite{jiangOptimalSpaceDepthTradeCNOT2022}.  The ancilla requirements in our RTS protocol vary depending on the specific algorithm being applied to. For BCCKS and ODE solver, the largest individual circuit (segment) does require a wider register; however, the total number of qubits that need to be prepared throughout is lower, because the average register size decreases. In algorithms such as QSP, where the ancilla count is independent of the truncation order, spatial cost is not affected by RTS.

RTS can also apply to the recently proposed LCHS~\cite{anLinearCombinationHamiltonian2023, an2023quantum} and QEP~\cite{lowQuantumEigenvalueProcessing2024} algorithms simulating non-unitary dynamics. Although no integer constraint applies to truncated-integral algorithms like LCHS, employing RTS also offers a quadratic improvement in truncation error. We anticipate that RTS can be further extended to encompass a broader range of quantum algorithms, including those involving time-dependent operators~\cite{lowHamiltonianSimulationInteraction2019} and infinite space truncations~\cite{mizutaOptimalHamiltonianSimulation2023, clinton2024towards}.

\section{Methods}
\label{sec:method}

\subsection{BCCKS}

In the BCCKS algorithm, we approximate the unitary, $U = e^{-iHt}$, through a truncated Taylor series, where each term in the series is unitary. One can apply the LCU algorithm \cite{childsHamiltonianSimulationUsing} to combine them to approximate $U$. 

More formally, any Hamiltonian $H$ can be represented by a sum of unitary components, i.e. 
\begin{equation}
    \label{eq:H_expand}
    H = \sum^L_{l=1}\alpha_l H_l.
\end{equation}

With Eq.~\eqref{eq:H_expand}, we express the order $K$ truncated Taylor series as 
\begin{equation}
    \label{eq:U_Taylor}
    \tilde{U} := \sum^K_{k=0} \frac{1}{k!}\left(-iHt\right)^k = \sum^K_{k=0} \sum^L_{l_1,\dots,l_k = 1} \overbrace{\frac{\alpha_{l_1} \alpha_{l_2} \ldots \alpha_{l_k} t^k}{k !}}^{\rm Coefficients }\underbrace{(-i)^k H_{l_1} H_{l_2} \ldots H_{l_k}}_{\rm Unitaries },
\end{equation}
where $\alpha_l \ge 0$ since we can absorb the negative sign in the corresponding $H_l$.
$\tilde{U}$ is in a standard form of LCU, i.e. $\sum_j \beta_j \tilde{V}_j$ for positive coefficients $\beta_j$ and unitaries $\tilde{V}_j$:

To implement the LCU algorithm, we first define two oracles
\begin{equation}
\label{eq:lcu_oracle}
 \begin{aligned}
G\ket{0} :=\frac{1}{\sqrt{s}}\sum_{j}\sqrt{\beta_j}\ket{j}, \quad
\mathbf{\rm SELECT}(\tilde{U}) :=\sum_{j} \ket{j}\bra{j} \otimes \tilde{V}_j,
\end{aligned}
\end{equation}
where $s = \sum_{j=0}^{L^K -1} \beta_j$.

With these two oracle, we can construct 
\begin{equation}
 \begin{aligned}
W := (G^{\dag}\otimes \mathds{1}) \mathbf{\rm SELECT}(\tilde{U}) (G\otimes \mathds{1}),
\end{aligned}
\end{equation}
where $\mathds{1}$ is the identity operator. Such that 
\begin{equation}
    W\ket{0}\ket{\psi} = \frac{1}{s}\ket{0}\tilde{U}\ket{\psi} + \ket{\perp},
\end{equation}
where $\left(\bra{0}\otimes\mathds{1}\right)\ket{\perp} = 0$.

Therefore, we successfully obtain $\tilde{U}\ket{\psi}$ heralding by measuring state $\ket{0}$ in the ancilla with probability $1/s^2$. Practically, the time of evolution $t$ is very large, making $1/s^2$ extremely small. We thus further apply oblivious amplitude amplification (OAA) \cite{berryExponentialImprovementPrecision2014} to amplify the success probability to near unity. To conduct, we control $s$ by dividing the evolution into $r = \left\lceil(\sum_{l=1}^L\alpha_l t) / \ln 2\right\rceil$ segments such that each segment has $s=2$ in the case of $K = \infty$. The last segment has a different $s$ due to the ceiling rounding, and its treatment is illustrated in ref. \cite{berryExponentialImprovementPrecision2014} with the cost of one additional ancillary qubit. In the actual implementation, since we only have finite $K$, $s$ will be slightly less than 2. However, OAA is robust as long as $|s -2| \le \mathcal{O}\left(\epsilon\right)$ and $\|\tilde{U} - U\| \le \mathcal{O}\left(\epsilon\right)$ as analyzed in ref. \cite{berrySimulatingHamiltonianDynamics2015}. By OAA, we can implement a segment with high probability, and one can approximate $\tilde{U}$ by concatenating $r$ segments. In the following discussion, we focus on just one segment, and the error and cost corresponding to the entire evolution can be retrieved by multiplying $r$.

After applying OAA, we have
\begin{equation}
\label{eq:sup_LCU_V1}
 \begin{aligned}
PTW (\ket{0}\otimes \ket{\Psi})=\ket{0}\otimes \left(\frac{3}{s}\tilde{U} - \frac{4}{s^3}\tilde{U}\tilde{U}^\dagger\tilde{U}  \right)\ket{\Psi} =: \ket{0}\otimes V_1\ket{\Psi},
\end{aligned}
\end{equation}
where $P := \ket{0}\bra{0}\otimes \mathds{1}$, $T = -WRW^{\dag}R$, and $R=(\mathds{1}-2\ket{0}\bra{0})\otimes \mathds{1}$. Therefore, we can construct a quantum circuit with post-selection, $V_1 = \left(\bra{0}\otimes\mathds{1}\right)TW\left(\ket{0}\otimes\mathds{1}\right)$, to approximates a $K_1$ truncated Eq.~\eqref{eq:U_Taylor}. 

Additionally, we define the index set $J_1$ for the mapping from $j \in J_1$ index to the tuple $(k,l_1,l_2,\dots,l_k)$ as 
\begin{equation}
    J_1 := \left\{(k,l_1,l_2,\dots,l_k): k\in \mathbb{N}, k\le K_1,l_1,l_2,\dots,l_k\in \{1,\dots,L\}\right\}.
\end{equation}
\begin{equation}
\label{
eq:LCU_V1}
\begin{aligned}
    \U &= \sum^{K_1}_{k=0}\frac{1}{k!}\left(-iH\tau\right)^k \\
    &= \sum^{K_1}_{k=0} \sum^L_{l_1,\dots,l_k = 1} \frac{\alpha_{l_1} \alpha_{l_2} \ldots \alpha_{l_k} t^k}{k !}(-i)^k H_{l_1} H_{l_2} \ldots H_{l_k} \\
    &= \sum_{j\in J_1}\beta_j\tilde{V}_j.
\end{aligned} 
\end{equation}
Thus, $\U$ can be implemented by invoking $G_1$ and ${\rm select}(\U)$ with the same definition in Eq.~\eqref{eq:lcu_oracle} with $j \in J_1$. 
The other quantum circuit $V_2$ implements the sum
\begin{equation}
\label{eq:LCU_V2}
\begin{aligned}
    \UU &= \sum^{K_1}_{k=0}\frac{1}{k!}\left(-iH\tau\right)^k + \frac{1}{1-p}\sum^{K_2}_{k=K_1+1}\frac{1}{k!}\left(-iH\tau\right)^k \\
    &= \sum^{K_1}_{k=0} \sum^L_{l_1,\dots,l_k = 1} \frac{\alpha_{l_1} \alpha_{l_2} \ldots \alpha_{l_k} t^k}{k !}(-i)^k H_{l_1} H_{l_2} \ldots H_{l_k}  + \frac{1}{1-p}\sum^{K_2}_{k=K_1+1} \sum^L_{l_1,\dots,l_k = 1} \frac{\alpha_{l_1} \alpha_{l_2} \ldots \alpha_{l_k} t^k}{k !}(-i)^k H_{l_1} H_{l_2} \ldots H_{l_k}\\
    &= \sum_{j\in J_2}\beta_j\tilde{V}_j,
\end{aligned} 
\end{equation}
Where $J_2 : = \left\{(k,l_1,l_2,\dots,l_k): k\in \mathbb{N}, k\le K_2,l_1,l_2,\dots,l_k\in \{1,\dots,L\}\right\}$.
Since we already set $\tau\sum^L_{l=1}\alpha_l = \ln 2$, and we must have $2p\delta_1/(1-p) \le a_2$ holds for bounding error in OAA, which can be easily violated as we increase $p$. We, hence, amplify $s_2 = \sum_{j\in J_2} \beta_j$ to $\sin(\pi/10))^{-1}$, and we will need one more flip to achieve approximately unit success probability. We also define $G_2$ and ${\rm select}(\UU)$ as in Eq.~\eqref{eq:lcu_oracle} with $j \in J_2$.

\subsection{QSP}

Here we derive the detailed error bounds for QSP-based transformations. QSP equips us with the tool for implementing non-linear combinations of Hamiltonians $H$, i.e. $f[H] = \sum_i \alpha_i H^i$ subject to some constraints on the coefficient $\alpha_i$. In this section, we will analyze how to apply RTS to two instances in QSP, namely HS and Uniform Spectral Amplification(USA). Oracles in QSP are similar to what was discussed in LCU. However, we express them in another form to be consistent with existing literature  \cite{lowHamiltonianSimulationQubitization2019a,childsHamiltonianSimulationUsing,lowQuantumSignalProcessing2017}. Note that all variables in this section are irrelevant to definitions in the last section.

Assuming we have two oracles: $\hat{G}$ prepares the state
\begin{equation}
    \hat{G} \ket{0}_a = \ket{G}_a \in \mathcal{H}_a,
\end{equation}
in the ancillary space $\mathcal{H}_a$ and $\hat{U}$ block encodes the Hamiltonian $H$ such that
\begin{equation}
    \begin{aligned}
        \hat{U}\ket{G}_a\ket{\lambda}_s &= \ket{G}_aH\ket{\lambda}_s + \sqrt{1 - \|H\ket{\lambda}\|^2}\ket{G^\perp}_{a}\ket{\lambda}_s \\
        &=\lambda\ket{G_\lambda}_{as} + \sqrt{1 - |\lambda|^2}\ket{G_\lambda^\perp}_{as},
    \end{aligned}
\end{equation}
where $\ket{\lambda}_s$ is one of the eigenstates of $H$ in the subspace $\mathcal{H}_s$ and $\lambda$ is the corresponding eigenvalue such that $H\ket{\lambda}_s = \lambda \ket{\lambda}_s$, $\ket{G_\lambda}_{as}$ is the abbreviation of $\ket{G}_a\ket{\lambda}_s$ and $(\bra{G}_a\otimes \mathds{\hat{1}})\ket{G_\lambda^\perp}_{as} = 0$. However, successively applying $\hat{U}$ does not produce powers of $\lambda$ because its action on $\ket{G_\lambda^\perp}_{as}$ contaminates the block we are interested in. We thus need to construct a unitary iterate $\hat{W}$  
\begin{equation}
    \begin{aligned}
        \hat{W} &= \left(\left(2\ket{G}\bra{G} - \mathds{\hat{1}}_a\right)\otimes \mathds{\hat{1}}_s\right)\hat{S}\hat{U}\\
        &= \begin{pmatrix}
        \lambda\ket{G_\lambda}\bra{G_\lambda} & -\sqrt{1-|\lambda|^2}\ket{G_\lambda}\bra{G_\lambda^\perp} \\
        \sqrt{1-|\lambda|^2}\ket{G_\lambda^\perp}\bra{G_\lambda} & \lambda\ket{G_\lambda^\perp}\bra{G_\lambda^\perp}
        \end{pmatrix},
    \end{aligned}
\end{equation}
where the construction of $\hat{S}$ is out of the scope of this article, and details can be found in ref.~\cite{lowHamiltonianSimulationQubitization2019a}. 
Note that from lemma 17 in ref.~\cite{childsQuantumAlgorithmSystems2017} power of $\hat{W}$ has the form
\begin{equation}
    \hat{W}^n = \begin{pmatrix}
\mathcal{T}_n(\lambda) & \cdot \\
\cdot & \cdot
\end{pmatrix},
\end{equation}
where $\mathcal{T}_n(\lambda)$ is the $n$-th order Chebyshev polynomial. However, $f[H]$ available for $\hat{W}$ is limited due to the restriction on parity. We can add an ancilla in subspace $\mathcal{H}_b$ to rotate $\hat{W}$ for a wider variety of $f[H]$. Define $\hat{V} = e^{i\Phi}\hat{W}$, $\hat{V}_0 = \ket{+}\bra{+}_b\otimes \mathds{\hat{1}}_s + \ket{-}\bra{-}_b\otimes\hat{V}$, and 
\begin{equation}
\label{eq:qsp_phase_iterate}
    \hat{V}_{\vec{\varphi}} = \prod^{K/2}_{\text{odd }k = 1}\hat{V}_{\varphi_{k+1}+\pi}^\dagger\hat{V}_{\varphi_{k}},
\end{equation}
where $\hat{V}_{\varphi} = \left(e^{-i\varphi \hat{Z}/2}\otimes \mathds{\hat{1}}_s\right)\hat{V}_0\left(e^{i\varphi \hat{Z}/2}\otimes \mathds{\hat{1}}_s\right)$ and $\hat{Z}\ket{\pm}_b = \ket{\mp}_b$.
Eq.~\eqref{eq:qsp_phase_iterate} essentially implements
\begin{equation}
\label{eq:qsp_general_ABCD}
    \hat{V}_{\vec{\varphi}} = \bigoplus_{\lambda,\pm} \left(\mathds{\hat{1}}_b\mathcal{A}(\theta_\lambda)+i\hat{Z}_b\mathcal{B}(\theta_\lambda) +i\hat{X}_b\mathcal{C}(\theta_\lambda) + i\hat{Y}_b\mathcal{D}(\theta_\lambda)\right)\otimes \ket{G_{\lambda\pm}}\bra{G_{\lambda\pm}}_{as},
\end{equation}
where $\ket{G_{\lambda\pm}} = (\ket{G_\lambda}\pm i\ket{G_\lambda^\perp}/\sqrt{2}$ and $(\mathcal{A},\mathcal{B},\mathcal{C},\mathcal{D})$ are real functions on $\theta_\lambda$.
We can classically solve for the vector $\vec{\varphi}$ to control each $\theta_\lambda$ thus implementing various functions of $H$. Since the constraints differ with parity of $f[H]$, we will specify the corresponding constraints along with functions to be implemented in HS and USA.

\begin{lemma}{(lemma 14 and Theorem 1 of ref.~\cite{lowHamiltonianSimulationQubitization2019a})
\label{le:qsp_functions_conditions}
For any even integer $Q > 0$, a choice of functions $\mathcal{A}(\theta)$ and $\mathcal{C}(\theta)$ is achievable by the framework of QSP if and only if the following are true:}
    \begin{enumerate}
        \item $\mathcal{A}(\theta) = \sum^{K}_{k=0} a_k\cos(k\theta)$ be a real cosine Fourier series of degree at most $K$, where $a_k$ are coefficients;
        \item $\mathcal{C}(\theta) = \sum^{K}_{k=1} c_k\sin(k\theta)$ be a real sine Fourier series of degree at most $K$, where $c_k$ are coefficients;
        \item $\mathcal{A}(0) = 1 + \epsilon_1$, where $|\epsilon_1| \le 1$;
        \item $\forall \theta \in \mathbb{R}, \mathcal{A}^2(\theta) + \mathcal{C}^2(\theta) \le 1 + \epsilon_2$, where $\epsilon_2 \in [0,1]$.
    \end{enumerate}
    Then, with $\tilde{\epsilon} = \epsilon_1 + \epsilon_2$, we can approximate the evolution unitary $e^{-iHt}$ with classically precomputed $\vec{\varphi} \in \mathbb{R}^{2K}$ such that
    \begin{equation}
        \left\|\bra{+}_b\bra{G}_a\hat{V}_{\vec{\varphi}}\ket{G}_a\ket{+}_b - e^{-iHt}\right\| \le \mathcal{O}\left(\sqrt{\tilde{\epsilon}}\right).
    \end{equation}
    The post-selection succeeds with probability at least $1-2\sqrt{\tilde{\epsilon}}$.
\end{lemma}
Note $\mathcal{A}(\theta)$ is maximized at $\theta = 0$, thus $\epsilon_1$ quantifies the truncation error of $\tilde{A}(\lambda)$.
Furthermore, $\epsilon_2$ quantifies errors introduced by rescaling, which will be necessary for the mixed-in term. As we amplify part of it by $1/(1-p)$, the magnitude of $\hat{V}_{\vec{\varphi}}$ may exceed 1 in such instance. 

There are two method to implement RTS in the context of QSP according to classical computation power. On one hand, QSP can realizes the Hamiltonian evolution in a single unitary, which corresponds to a long evolution time. Therefore, we need to truncate Eq.~\eqref{eq:qsp_ac} at a very high $K$ as it scales linear with evolution time. Consequently, the calculation of $\vec{\varphi}$ becomes extremely hard. On the other hand, we can splits the evolution time into segments and concatenates these short-time evolutions to approximate the overall Hamiltonian dynamics thus the classical calculation becomes easier. However, this method faces the problem of large constant overhead in quantum resources. RTS embraces both cases. We assume using segmentized QSP for simplicity, and the analysis for the other is a straightforward extension of our result. In the following discussion, We will give constructions of the two unitaries we are mixing. We drop the subscript $\vec{\varphi}$ in $\hat{V}_{\vec{\varphi}}$ for simplicity, and denote $\hat{V}_1$ and $\hat{V}_2$ as the two unitaries, where

\begin{equation}
\label{eq:formal_qsp_HS}
    \begin{aligned}
        \hat{V}_1(\lambda) = &A_1(\lambda) + iC_1(\lambda) \\
        = &J_0(t) + 2\sum^{K_1}_{\text{even } k>0}(-1)^{\frac{k}{2}}J_k(t)T_k(\lambda) + i2\sum^{K_1}_{\text{odd } k>0}(-1)^{\frac{k-1}{2}}J_k(t)T_k(\lambda) \\
        \hat{V}_2(\lambda) =  &A_2(\lambda) + iC_2(\lambda) \\
        &J_0(t) + 2\sum^{K_1}_{\text{even } k>0}(-1)^{\frac{k}{2}}J_k(t)T_k(\lambda) +  \frac{1}{1-p}\left(2\sum^{K_2}_{\text{even } k>K_1}(-1)^{\frac{k}{2}}J_k(t)T_k(\lambda)\right)\\
        &+ i2\sum^{K_1}_{\text{odd } k>0}(-1)^{\frac{k-1}{2}}J_k(t)T_k(\lambda) +i\frac{1}{1-p}\left(2\sum^{K_2}_{\text{odd } k>K_1}(-1)^{\frac{k-1}{2}}J_k(t)T_k(\lambda)\right).
    \end{aligned}
\end{equation}

Further, the virtual operator appears in the error analysis is 
\begin{equation}
    \begin{aligned}
        \hat{V}_m(\lambda) &= p\hat{V}_1(\lambda) + (1-p) \hat{V}_2(\lambda) \\
        &= \left(pA_1(\lambda) + (1-p)A_2(\lambda)\right) + i\left(pC_1(\lambda) + (1-p)C_2(\lambda)\right)\\
        &=: A_m(\lambda) + i C_m(\lambda).
    \end{aligned}
\end{equation}

We approximate $\{A_i,C_i| i \in \{1,2,m\}\}$ by the corresponding rescaled operators $\{\tilde{A}_i,\tilde{C}_i| i \in \{1,2,m\}\}$, and the error is quantified by lemma \ref{le:qsp_functions_conditions}. We thus need to calculate the corresponding $\epsilon_1$ and $\epsilon_2$ for all three operators.

\subsubsection{Hamiltonian Simulation}
Given $\vec{\varphi}\in \mathbb{R}^K$, by choosing $\Phi = \pi/2$ and projecting $\hat{V}_{\vec{\varphi}}\ket{+}_b\ket{G}_a$ on to $\bra{G}_a\bra{+}_b$ with post-selection, Eq.~\eqref{eq:qsp_general_ABCD} becomes
\begin{equation}
    \label{eq:QSP_HS_phase_iterate}
    \bra{G}_a\bra{+}_b\hat{V}_{\vec{\varphi}}\ket{+}_b\ket{G}_a = \bigoplus_{\lambda} \left(\tilde{A}(\lambda)+i\tilde{C}(\lambda)\right)\otimes \ket{\lambda}\bra{\lambda}_s,
\end{equation}
where $\tilde{A}(\lambda) = \sum^{K/2}_{\text{even } k=0}a_kT_k(\lambda)$, $\tilde{C}(\lambda) = \sum^{K/2}_{\text{odd } k=1}c_kT_k(\lambda)$, and $a_k, c_k\in \mathbb{R}$ are coefficients depending on $\vec{\varphi}$.

$U$ can be decomposed using the Jacobi-Anger expansion \cite{abramowitz1988handbook}
\begin{equation}
\label{eq:qsp_decompose}
\begin{aligned}
    e^{-i\lambda t }&= J_0(t) + 2\sum^\infty_{{\rm even}\, k>0}(-1)^{\frac{k}{2}}J_k(t)T_k(\lambda) + i2\sum^\infty_{{\rm odd}\,  k>0}(-1)^{\frac{k-1}{2}}J_k(t)T_k(\lambda)\\
    &= A(\lambda) + iC(\lambda),
\end{aligned}
\end{equation}
where $J_k$ is the Bessel function of the first kind and $T_k$ are Chebyshev's polynomials. We truncate $A(\lambda)$ and $C(\lambda)$ in Eq.~\eqref{eq:qsp_decompose} at order $K$, such that we need to compute $\vec{\varphi}_1\in \mathbb{R}^{2K}$ classically to obtain
\begin{equation}
\label{eq:qsp_ac}
    \begin{gathered}
        \tilde{A}(\lambda) =  J_0(t) + 2\sum^{K}_{\text{even } k>0}(-1)^{\frac{k}{2}}J_k(t)T_k(\lambda) \\
        \tilde{C}(\lambda) = 2\sum^{K}_{\text{odd } k>0}(-1)^{\frac{k-1}{2}}J_k(t)T_k(\lambda).
    \end{gathered}
\end{equation}

We will need the $\tilde{A}(\lambda)$ and $\tilde{C}(\lambda)$ to satisfy constraints in lemma \ref{le:qsp_functions_conditions}, which gives how robust QSP is when approximating $A(\lambda), C(\lambda)$ by $\tilde{A}(\lambda),\tilde{C}(\lambda)$.

\subsubsection{Uniform spectral amplification (USA)}
Going back to Eq.~\eqref{eq:qsp_general_ABCD}, with another ancilla qubit in $\mathcal{H}_c$, we can define 
\begin{equation}
    \hat{W}_{\vec{\varphi}} = \hat{V}_{\vec{\varphi}}\otimes \ket{+}\bra{+}_c + \hat{V}_{\pi - \vec{\varphi}}\otimes \ket{-}\bra{-}_c.
\end{equation}

We then can project $\hat{W}_{\vec{\varphi}}\ket{G}_a\ket{0}_b\ket{0}_c$ onto $\bra{0}_c\ket{0}_b\ket{G}_a$ such that 
\begin{equation}
\label{eq:APP_qsp_hs_qc}
    \bra{0}_c\bra{0}_b\bra{G}_a\hat{W}_{\vec{\varphi}}\ket{G}_a\ket{0}_b\ket{0}_c = D(\lambda)\otimes \ket{\lambda}\bra{\lambda},
\end{equation}
where $D$ is an odd real polynomial function of degree at most $2K+1$ satisfying $\forall \lambda \in [-1,1]$, $D^2(\lambda) \le 1$. The rescaling in this case becomes easy because we can neglect $A, B$ and $C$ in Eq.~\eqref{eq:qsp_general_ABCD}. For the mix-in term, where the norm may be greater than 1, we can simply rescale it by a constant factor, and the upper bound on the corresponding error will be doubled. 

In the task of USA, we would like to approximate the truncated linear function 
\begin{equation}
\label{eq:QSP_TLF}
    f_{\Gamma,\delta}(\lambda) = \begin{cases}
        \frac{\lambda}{2\Gamma}, & |\lambda| \in [0,\Gamma] \\
        \in [-1,1], &  |\lambda| \in (\Gamma,1],
    \end{cases}
\end{equation}
where $\delta = \displaystyle\max_{|x| \in [0,\Gamma]} \left|\frac{|x|}{2\Gamma}\tilde{f}_\Gamma(\lambda) - 1\right|$ is the maximum error tolerance. 

We can approximate Eq.~\eqref{eq:QSP_TLF} by 
\begin{equation}
\label{eq:QSP_tlf_construction}
    f_{\Gamma,\delta}(\lambda ) = \frac{\lambda}{4\Gamma}\left(erf\left(\frac{\lambda +2\Gamma}{\sqrt{2}\Gamma \delta^\prime}\right)+erf\left(\frac{2\Gamma-\lambda }{\sqrt{2}\Gamma \delta^\prime}\right)\right),
\end{equation}
where $1/\delta^\prime = \sqrt{\log (2/(\pi\delta^2)}$ and $erf(\gamma x) = \frac{2}{\pi}\int^{\gamma x}_0 e^{-t^2}dt = \frac{2}{\sqrt{\pi}}\int^{x}_0 e^{-(\gamma t)^2}dt$ is the error function. Observe that we can approximate the truncated linear function by a combination of error functions only. The following is the construction of the error function by Chebyshev's polynomial using the truncated Jacobi-Anger expansion

\begin{equation}
\label{eq:erf}
    P_{erf,\gamma ,K_1}(\lambda) = \frac{2\gamma e^{-\gamma ^2/2}}{\sqrt{\pi}}\left(J_0\left(\frac{\gamma ^2}{2}\right)\lambda+\sum^{(K_1-1)/2}_{k=1}J_k\left(\frac{\gamma ^2}{2}\right)(-1)^k\left(\frac{T_{2k+1}(\lambda)}{2k+1}-\frac{T_{2k-1}(\lambda)}{2k-1}\right)\right).
\end{equation}
The polynomial we mixed in with small probability is 
\begin{equation}
    \label{eq:erf_2}
    \begin{aligned}
    P_{erf,\gamma ,K_2}(\lambda) =& \frac{2\gamma e^{-\gamma ^2/2}}{\sqrt{\pi}}\left(J_0\left(\frac{\gamma ^2}{2}\right)\lambda  +\left(\sum^{(K_1-1)/2}_{k=1}J_k\left(\frac{\gamma ^2}{2}\right)(-1)^k\left(\frac{T_{2k+1}(\lambda)}{2k+1}-\frac{T_{2k-1}(\lambda)}{2k-1}\right)\right)\right. \\
    & \qquad \qquad \qquad \qquad \; \;\quad + \left.\frac{1}{1-p}\left(\sum^{(K_2-1)/2}_{k=(K_1+1)/2}J_k\left(\frac{\gamma ^2}{2}\right)(-1)^k\left(\frac{T_{2k+1}(\lambda)}{2k+1}-\frac{T_{2k-1}(\lambda)}{2k-1}\right)\right)\right),
    \end{aligned}
\end{equation}

We can then substitute Eq.~\eqref{eq:erf} and~\eqref{eq:erf_2} into Eq.~\eqref{eq:QSP_tlf_construction} to approximate truncated linear function, which results in $\hat{P}_{\Gamma,\delta,K_{1(2)}}$ for replacing $erf(\lambda)$ by $P_{erf,\gamma ,K_{1(2)}}(\lambda)$
\begin{equation}
\label{ep:QSP_lin_v_1(2)}
     \hat{P}_{\Gamma,\delta,K_{1(2)}}(\lambda) = \frac{\lambda}{4\Gamma}\left(P_{erf,\gamma ,K_{1(2)}}\left(\frac{\lambda +2\Gamma}{\sqrt{2}\Gamma \delta^\prime}\right)+P_{erf,\gamma ,K_{1(2)}}\left(\frac{2\Gamma-\lambda }{\sqrt{2}\Gamma \delta^\prime}\right)\right).
\end{equation}

Follow the error propagation in ref.  \cite{lowQuantumSignalProcessing2017}, we can bound 
\begin{equation}
\label{eq:QSP_lin_erf_bound_transfer}
    \epsilon_{\Gamma,K} = \max_{|\lambda| \in [0,\Gamma]}\frac{2\Gamma}{|\lambda|} \left|\hat{P}_{\Gamma,\delta,K}(\lambda)-\frac{\lambda}{2\Gamma}\right|\le 2\epsilon_{erf,4\Gamma,K-1},
\end{equation}
where $\epsilon_{erf,\Gamma,K}$ is the truncation error of  $P_{erf,\gamma ,K}(\lambda)$ approximating $erf(\gamma\lambda)$.

\subsection{ODE}

Consider a differential equation of the form 
\begin{equation}
    \frac{d\vec{x}}{dt} = A\vec{x}+ \vec{b},
\end{equation}
where $A \in \mathbb{R}^{n\times n}$, $\vec{b} \in \mathbb{R}^n$ are time-independent. The exact solution is given by
\begin{equation}
    \vec{x}(t) = e^{At}\vec{x}(0) + \left(e^{At}-\mathds{1}_n\right)A^{-1}\vec{b},
\end{equation}
where $\mathds{1}_n$ is the $n$-dimensional identity vector. 

We can approximate $e^{z}$ and $\left(e^{z}-\mathds{1}_n\right)z^{-1}$ by two $k$-truncated Taylor expansions:
\begin{equation}
\label{eq:APP_Tk}
    T_k(z) := \sum^{K}_{k=0} \frac{z^k}{k!} \approx e^{z}
\end{equation}
and
\begin{equation}
\label{eq:APP_Sk}
    S_k(z) := \sum^K_{k=1}\frac{z^{k-1}}{k!} \approx \left(e^{z}-1\right)z^{-1}.
\end{equation}

Consider a short time $h$, We can approximate the solution $\vec{x}(qh)$ recursively from $\vec{x}((q-1)h)$, for an integer $q$. Denote $x^q$ as the solution approximated by the algorithm, we have  
\begin{equation}
    x^{q+1} = T_k(Ah)x^{q} + S_k(Ah)h\vec{b}.
\end{equation}
Furthermore, we can embed the series of recursive equations into a large linear system $\mathcal{L}$ as proposed in  \cite{berryQuantumAlgorithmLinear2017} such that the solution to $\mathcal{L}$ gives the history state \cite{berryHighorderQuantumAlgorithm2014} of $x$, which encodes solution at all time steps. $\mathcal{L}$ has the form
\begin{equation}
\label{eq:ODE_LS}
    C_{m,K,p}(Ah) \ket{x} = \ket{0}\ket{x_{in}} + h\sum^{m-1}_{i=0}\ket{i(K+1)+1}\ket{b},
\end{equation}
where $m$ is the maximum time step and $p$ is the repetition number of identity operator after evolution aiming to increase the probability of projecting onto the final state. The operator has the form
\begin{equation}
    \begin{split}
    C_{m,K_1,p} (A) := &\sum^{d_1}_{j=0} \ket{j}\bra{j}\otimes \mathds{1} - \sum^{m-1}_{i=0}\sum^{K_1}_{j=1}\ket{i(K_1+1)+j}\bra{i(K_1+1)+j-1}\otimes \frac{A}{j}\\
    &-\sum^{m-1}_{i=0}\sum^{K_1}_{j=0}\ket{(i+1)(K_1+1)}\bra{i(K_1+1)+j}\otimes \mathds{1}-\sum^d_{j=d-p+1}\ket{j}\bra{j-1}\otimes \mathds{1},
    \end{split}
\end{equation}
where $d_1 = m(K_1+1) + p$. To implement RTS, we need to apply the modified higher order terms in Eq.~\eqref{eq:APP_Tk} and~\eqref{eq:APP_Sk}. This could be done by performing the operator
\begin{equation}
    \begin{split}
    \tilde{C}_{m,K_1,K_2,p}(A) := &\sum^{d_2}_{j=0} \ket{j}\bra{j} \otimes \mathds{1} - \sum^{m-1}_{i=0}\sum^{K_2}_{j=1}\ket{i(K_2+1)+j}\bra{i(K_2+1)+j-1}\otimes \frac{A}{j} \\
    &-\frac{p}{(1-p)}\sum^{m-1}_{i=0}\ket{i(K_2+1)+K_1+1}\bra{i(K_2+1)+K_1}\otimes \frac{A}{j}\\
    &-\sum^{m-1}_{i=0}\sum^{K_2}_{j=0}\ket{(i+1)(K_2+1)}\bra{i(K_2+1)+j}\otimes \mathds{1}-\sum^d_{j=d-p+1}\ket{j}\bra{j-1}\otimes \mathds{1},
    \end{split}
\end{equation}
where $d_2 = m(K_2+1)+p$.

$\mathcal{L}$ with $C_{m,k,p}(Ah)$ and $\tilde{C}_{m,K_1,K_2,p}(Ah)$ gives solutions $\ket{x}$ and $\ket{\tilde{x}}$ respectively, where
\begin{itemize}
    \item $\ket{x_{i,j}}$ satisfies
            \begin{equation}
                \begin{aligned}
                    &\ket{x_{0,0}} = \ket{x_{in}}, &&\\
                    &\ket{x_{i,0}} = \sum^{K_1}_{j=0}\ket{x_{i-1,j}}, && 1\le i \le m \\
                    &\ket{x_{i,1}} = Ah\ket{x_{i,0}}+h\ket{b}, && 0\le i < m \\
                    &\ket{x_{i,j}} = \frac{Ah}{j}\ket{x_{i,j-1}}, &&0 \le i <m, 2\le j \le K_1\\
                    &\ket{x_{m,j}} = \ket{x_{m,j_1}}, &&1\le j \le p
                \end{aligned}
            \end{equation}
    \item $\ket{\tilde{x}_{i,j}}$ satisfies 
            \begin{equation}
                \begin{aligned}
                    &\ket{\tilde{x}_{0,0}} = \ket{\tilde{x}_{in}}, &&\\
                    &\ket{\tilde{x}_{i,0}} = \sum^{K_2}_{j=0}\ket{\tilde{x}_{i-1,j}}, && 1\le i \le m \\
                    &\ket{\tilde{x}_{i,1}} = Ah\ket{\tilde{x}_{i,0}}+h\ket{b}, && 0\le i < m \\
                    &\ket{\tilde{x}_{i,j}} = \frac{Ah}{j}\ket{\tilde{x}_{i,j-1}}, &&0 \le i <m, 2\le j \le K_1\\
                    &\ket{\tilde{x}_{i,j}} = \frac{1}{(1-p)^{1/(K_2-K_1)}}\frac{Ah}{j}\ket{\tilde{x}_{i,j-1}}, &&0 \le i <m, K_1+1\le j \le K_2\\
                    &\ket{\tilde{x}_{m,j}} = \ket{\tilde{x}_{m,j_1}}, &&1\le j \le p
                \end{aligned}
            \end{equation} 
\end{itemize}

Therefore,
\begin{equation}
    \ket{x_{m,j}} = \tilde{T}_{K_2}(Ah)\ket{x_{m_1},0}+\tilde{S}_{K_2}(Ah)h\ket{b},
\end{equation}
where
\begin{equation}
    \begin{aligned}
        \tilde{T}_{K_2} &= \sum^{K_1}_{k=0} \frac{z^k}{k!} + \frac{1}{1-p}\sum^{K_2}_{k=K_1+1} \frac{z^k}{k!} \\
        \tilde{S}_{K_2} &= \sum^{K_1}_{k=1} \frac{z^{k-1}}{k!} + \frac{1}{1-p}\sum^{K_2}_{k=K_1+1} \frac{z^{k-1}}{k!}
    \end{aligned}
\end{equation}

Our quantum circuit solves the linear system described by Eq.~\eqref{eq:ODE_LS}. However, quantifying the operator norm between $C^{-1}_{m,k,p}$ and an errorless $C_\infty^{-1}$, which is Taylor series summed to infinity order, is meaningless because 1. $C^{-1}_{m,k,p}$ and $C_\infty^{-1}$ has different dimension. 2. The difference between them does not directly reflect the distance of the final state obtained by the algorithm. We therefore choose to follow the derivation in ref.  \cite{berryQuantumAlgorithmLinear2017} that evaluate the state distance between $x^m$ and $x(mh)$. Equivalently, we are measuring the distance between the post-selected operator and an idea evolution on the target state. We will first prove the operator norm distance for $T_k$, followed by proving these distances can also upper bound the corresponding norm for $S_k$

\begin{proof}

Let us first consider the truncation errors concerning $T_k$.   
We are going to derive $\alpha_1 = \|e^{z} - T_{K_1}\|,\alpha_2 = \|e^{z} - T_{K_2}^\prime\|$, and $\beta = \| e^z - \left(pT_{K_1} + (1-p) T_{K_2}^\prime\right)\|$ for $T_k$, where 
\begin{equation}
\begin{aligned}
    T_{K_1} &= \sum^{K_1}_{k=0} \frac{z^k}{k!}\\
    \tilde{T}_{K_2} &= \sum^{K_1}_{k=0} \frac{z^k}{k!} + \frac{1}{1-p}\sum^{K_2}_{k=K_1+1} \frac{z^k}{k!}
\end{aligned}
\end{equation}

\begin{enumerate}
\item  From the proof of Lemma \ref{le:lcu_v1}, we can write $\alpha_1 \le  2\frac{\left(\ln 2\right)^{K_1+1}}{(K_1+1)!}$
\item Similarly, from Lemma \ref{le:lcu_v2}, we can write $\alpha_2 \le \left|-\frac{p}{1-p}\frac{(\ln 2)^{K_1+1}}{(K_1+1)!}+ \frac{2}{1-p} \frac{(\ln 2)^{K_2+1}}{(K_2+1)!}\right|\approx \frac{p}{1-p}\frac{(\ln 2)^{K_1+1}}{(K_1+1)!} $
\item Observe that $pT_{K_1} + (1-p) T_{K_2}^\prime = T_{k_2}$, $\beta \le 2\frac{\left(\ln 2\right)^{K_2+1}}{(K_2+1)!}$
\end{enumerate}

As for $S_k$, we will prove three bounds on  $\alpha^\prime_1 = \|\left(e^z-1\right)z^{-1} - S_{K_1}\|,\alpha^\prime_2 = \|\left(e^z-1\right)z^{-1} - S_{K_2}^\prime\|$, and $\beta^\prime = \| \left(e^z-1\right)z^{-1} - \left(pS_{K_1} + (1-p) S_{K_2}^\prime\right)\|$,
where 
\begin{equation}
\begin{aligned}
    S_{K_1} &= \sum^{K_1}_{k=1} \frac{z^{k-1}}{k!}\\
   \tilde{S}_{K_2} &= \sum^{K_1}_{k=1} \frac{z^{k-1}}{k!} + \frac{1}{1-p}\sum^{K_2}_{k=K_1+1} \frac{z^{k-1}}{k!}
\end{aligned}
\end{equation}
Consider 
\begin{equation}
\begin{aligned}
    &e^z - T_k(z) \\
    =&\left(e^z - 1\right) - \left(T_k(z)-1\right) \\
    =& \left(e^z - 1\right) - zS_k(z)
\end{aligned}
\end{equation}
Therefore, with $|z| \le 1$, $\|e^{z} - T_{K_1}\| = \|\left(\left(e^z-1\right) - zS_{K_1}\right)\| = \|z\left(\left(e^z-1\right)z^{-1} - S_{K_1}\right)\| \ge \|\left(e^z-1\right)z^{-1} - S_{K_1}\|$. $S_k$ and $T_k$ share the same bounds. 

Denote $x_{m}^j$ as the state obtained by solving the linear system defined by $C_{m,K_2,p}(A)$.
We have $\delta_1 = \|\ket{x(jh)} - x_1^{j}\|$, $\delta_2 = \|\ket{x(jh)} - x_2^{j}\|$ and $\delta_m = \|\ket{x(jh)} - x_{m}^j\|$. 

By inserting bounds on $\alpha_1,\alpha_2$ and $\beta$ into the proof of Theorem 6 in Ref.~\cite{berryQuantumAlgorithmLinear2017}, we obtain that
\begin{equation}
    \delta_1 \le \frac{\mathcal{C}_j}{(K_1+1)!} \quad \delta_2 \le \frac{p}{1-p}\frac{\mathcal{C}_j}{(K_1+1)!} \quad \delta_m \le \frac{\mathcal{C}_j}{(K_2+1)!},
\end{equation}
where $\mathcal{C}_j = 2.8 \kappa_V j \left(\|\ket{x}_{in}\|+ mh \|\ket{b}\|\right)$, and $\kappa_V$ is the condition number in the eigendecomposition of $C_{m,K,p} = VDV^{-1}$ for some diagonal matrix $D$. Further applying lemma \ref{le:mix}, we can bound the stated state distance by.
\end{proof}

\textit{Note added:} After the submission of this manuscript, we became aware of a related work~\cite{martyn2025halving} that appeared subsequently and addresses a similar problem of reducing the cost of quantum algorithms using randomization. Our study was developed independently, and the main results reported here were obtained prior to the appearance of that work. Although both approaches employ randomization, the constructions are fundamentally different. The related work mixes over an ensemble of truncated polynomials within the framework of quantum signal processing (QSP), which requires numerically determining $\mathcal{O}(K)$ sets of phase angles for implementation, leading to circuit structures that may vary substantially. In contrast, our method involves only two polynomials, offering a simpler and more transparent framework that facilitates both classical optimization and experimental realization.

\subsection*{Funding}
    Q.Z. is supported by funding from Innovation Program for Quantum Science and Technology via Project 2024ZD0301900, National Natural Science Foundation of China (NSFC) via Project No. 12347104 and No. 12305030, Guangdong Basic and Applied Basic Research Foundation via Project 2023A1515012185, Hong Kong Research Grant Council (RGC) via No. 27300823, N\_HKU718/23, and R6010-23, Guangdong Provincial Quantum Science Strategic Initiative No. GDZX2303007, HKU Seed Fund for Basic Research for New Staff via Project 2201100596. 

\subsection*{Acknowledgements}
    Q.Z. acknowledges funding from Quantum Science and Technology-National Science and Technology Major Project 2024ZD0301900, National Natural Science Foundation of China (NSFC) via Project No. 12347104 and No. 12305030, Guangdong Basic and Applied Basic Research Foundation via Project 2023A1515012185, Hong Kong Research Grant Council (RGC) via No. 27300823, N\_HKU718/23, and R6010-23, Guangdong Provincial Quantum Science Strategic Initiative No. GDZX2303007.
    
\subsection*{Author Contributions}
Y.W and Q.Z collaboratively initialized the research ideas, derived the mathematical proofs and wrote the manuscript. Y.W. performed the numerical simulations.

\subsection*{Competing Interests}
The authors declare no competing interests.

\bibliography{bibli}

@article{berrySimulatingHamiltonianDynamics2015,
  title = {Simulating {{Hamiltonian Dynamics}} with a {{Truncated Taylor Series}}},
  author = {Berry, Dominic W. and Childs, Andrew M. and Cleve, Richard and Kothari, Robin and Somma, Rolando D.},
  year = {2015},
  month = mar,
  journal = {Physical Review Letters},
  volume = {114},
  number = {9},
  pages = {090502},
  publisher = {{American Physical Society}},
  doi = {10.1103/PhysRevLett.114.090502},
  urldate = {2023-07-31}
}

@article{lowHamiltonianSimulationQubitization2019a,
  title = {Hamiltonian {{Simulation}} by {{Qubitization}}},
  author = {Low, Guang Hao and Chuang, Isaac L.},
  year = {2019},
  month = jul,
  journal = {Quantum},
  volume = {3},
  pages = {163},
  publisher = {{Verein zur F\"orderung des Open Access Publizierens in den Quantenwissenschaften}},
  doi = {10.22331/q-2019-07-12-163},
  urldate = {2023-07-31},
  langid = {british}
}

@article{childsHamiltonianSimulationUsing,
  title = {Hamiltonian {{Simulation Using Linear Combinations}} of {{Unitary Operations}}},
  author = {Childs, Andrew M. and Wiebe, Nathan},
  journal = {Quantum Information and Computation},
  year = {2012},
  volume = {12},
  number = {11\&12},
  eprint = {1202.5822},
  primaryclass = {quant-ph},
  issn = {15337146, 15337146},
  doi = {10.26421/QIC12.11-12},
  urldate = {2023-07-31},
  archiveprefix = {arxiv}
}

@inproceedings{berryExponentialImprovementPrecision2014,
  title = {Exponential Improvement in Precision for Simulating Sparse {{Hamiltonians}}},
  booktitle = {Proceedings of the Forty-Sixth Annual {{ACM}} Symposium on {{Theory}} of Computing},
  author = {Berry, Dominic W. and Childs, Andrew M. and Cleve, Richard and Kothari, Robin and Somma, Rolando D.},
  year = {2014},
  month = may,
  eprint = {1312.1414},
  primaryclass = {quant-ph},
  pages = {283--292},
  doi = {10.1145/2591796.2591854},
  urldate = {2023-08-15},
  archiveprefix = {arxiv}
}

@article{childsFirstQuantumSimulation2018,
  title = {Toward the First Quantum Simulation with Quantum Speedup},
  author = {Childs, Andrew M. and Maslov, Dmitri and Nam, Yunseong and Ross, Neil J. and Su, Yuan},
  year = {2018},
  month = sep,
  journal = {Proceedings of the National Academy of Sciences},
  volume = {115},
  number = {38},
  pages = {9456--9461},
  publisher = {{Proceedings of the National Academy of Sciences}},
  doi = {10.1073/pnas.1801723115},
  urldate = {2023-08-16},
  langid = {english}
}

@article{childsQuantumAlgorithmSystems2017,
  title = {Quantum Algorithm for Systems of Linear Equations with Exponentially Improved Dependence on Precision},
  author = {Childs, Andrew M. and Kothari, Robin and Somma, Rolando D.},
  year = {2017},
  month = jan,
  journal = {SIAM Journal on Computing},
  volume = {46},
  number = {6},
  eprint = {1511.02306},
  primaryclass = {quant-ph},
  pages = {1920--1950},
  issn = {0097-5397, 1095-7111},
  doi = {10.1137/16M1087072},
  urldate = {2022-11-26},
  archiveprefix = {arxiv}
}

@misc{abramowitz1988handbook,
  title={Handbook of mathematical functions with formulas, graphs, and mathematical tables},
  author={Abramowitz, Milton and Stegun, Irene A and Romer, Robert H},
  year={1988},
  publisher={American Association of Physics Teachers}
}

@article{martynGrandUnificationQuantum2021,
  title = {A {{Grand Unification}} of {{Quantum Algorithms}}},
  author = {Martyn, John M. and Rossi, Zane M. and Tan, Andrew K. and Chuang, Isaac L.},
  year = {2021},
  month = dec,
  journal = {PRX Quantum},
  volume = {2},
  number = {4},
  eprint = {2105.02859},
  primaryclass = {quant-ph},
  pages = {040203},
  issn = {2691-3399},
  doi = {10.1103/PRXQuantum.2.040203},
  urldate = {2023-05-11},
  archiveprefix = {arxiv}
}

@misc{liuDenseOutputsQuantum2023,
  title = {Dense Outputs from Quantum Simulations},
  author = {Liu, Jin-Peng and Lin, Lin},
  year = {2023},
  month = jul,
  number = {arXiv:2307.14441},
  eprint = {2307.14441},
  primaryclass = {quant-ph},
  publisher = {{arXiv}},
  doi = {10.48550/arXiv.2307.14441},
  urldate = {2023-10-11},
  archiveprefix = {arxiv}
}

@phdthesis{lowQuantumSignalProcessing2017,
  type = {Thesis},
  title = {Quantum Signal Processing by Single-Qubit Dynamics},
  author = {Low, Guang Hao},
  year = {2017},
  urldate = {2023-11-15},
  copyright = {MIT theses are protected by copyright. They may be viewed, downloaded, or printed from this source but further reproduction or distribution in any format is prohibited without written permission.},
  langid = {english},
  school = {Massachusetts Institute of Technology}
}

@article{berryQuantumAlgorithmLinear2017,
  title = {Quantum Algorithm for Linear Differential Equations with Exponentially Improved Dependence on Precision},
  author = {Berry, Dominic W. and Childs, Andrew M. and Ostrander, Aaron and Wang, Guoming},
  year = {2017},
  month = dec,
  journal = {Communications in Mathematical Physics},
  volume = {356},
  number = {3},
  eprint = {1701.03684},
  primaryclass = {quant-ph},
  pages = {1057--1081},
  issn = {0010-3616, 1432-0916},
  doi = {10.1007/s00220-017-3002-y},
  urldate = {2023-09-30},
  archiveprefix = {arxiv}
}

@article{SommaSpectral2013,
author = {Somma, R. D. and Boixo, S.},
title = {Spectral Gap Amplification},
journal = {SIAM Journal on Computing},
volume = {42},
number = {2},
pages = {593-610},
year = {2013},
doi = {10.1137/120871997},

URL = { 
        https://doi.org/10.1137/120871997
},
eprint = { 
        https://doi.org/10.1137/120871997
}
}

@misc{nagaj2009fast,
      title={Fast Amplification of QMA}, 
      author={Daniel Nagaj and Pawel Wocjan and Yong Zhang},
      year={2009},
      eprint={0904.1549},
      archivePrefix={arXiv},
      primaryClass={quant-ph}
}

@article{berryHighorderQuantumAlgorithm2014,
  title = {High-Order Quantum Algorithm for Solving Linear Differential Equations},
  author = {Berry, Dominic W.},
  year = {2014},
  month = mar,
  journal = {Journal of Physics A: Mathematical and Theoretical},
  volume = {47},
  number = {10},
  eprint = {1010.2745},
  primaryclass = {quant-ph},
  pages = {105301},
  issn = {1751-8113, 1751-8121},
  doi = {10.1088/1751-8113/47/10/105301},
  urldate = {2023-12-12},
  archiveprefix = {arxiv}
}

@article{campbellShorterGateSequences2017,
  title = {Shorter Gate Sequences for Quantum Computing by Mixing Unitaries},
  author = {Campbell, Earl},
  year = {2017},
  month = apr,
  journal = {Physical Review A},
  volume = {95},
  number = {4},
  pages = {042306},
  publisher = {{American Physical Society}},
  doi = {10.1103/PhysRevA.95.042306},
  urldate = {2023-12-27}
}

@misc{hastingsTurningGateSynthesis2016,
  title = {Turning {{Gate Synthesis Errors}} into {{Incoherent Errors}}},
  author = {Hastings, M. B.},
  year = {2016},
  month = dec,
  journal = {arXiv.org},
  urldate = {2024-01-09},
  howpublished = {https://arxiv.org/abs/1612.01011v1},
  langid = {english}
}

@misc{liPotentialQuantumAdvantage2023,
  title = {Potential Quantum Advantage for Simulation of Fluid Dynamics},
  author = {Li, Xiangyu and Yin, Xiaolong and Wiebe, Nathan and Chun, Jaehun and Schenter, Gregory K. and Cheung, Margaret S. and M{\"u}lmenst{\"a}dt, Johannes},
  year = {2023},
  month = apr,
  number = {arXiv:2303.16550},
  eprint = {2303.16550},
  primaryclass = {physics, physics:quant-ph},
  publisher = {{arXiv}},
  doi = {10.48550/arXiv.2303.16550},
  urldate = {2024-01-11},
  archiveprefix = {arxiv}
}

@article{lindenQuantumVsClassical2022,
  title = {Quantum vs. {{Classical Algorithms}} for {{Solving}} the {{Heat Equation}}},
  author = {Linden, Noah and Montanaro, Ashley and Shao, Changpeng},
  year = {2022},
  month = oct,
  journal = {Communications in Mathematical Physics},
  volume = {395},
  number = {2},
  pages = {601--641},
  issn = {1432-0916},
  doi = {10.1007/s00220-022-04442-6},
  urldate = {2024-01-11},
  langid = {english}
}

@article{costaQuantumAlgorithmSimulating2019,
  title = {Quantum Algorithm for Simulating the Wave Equation},
  author = {Costa, Pedro C. S. and Jordan, Stephen and Ostrander, Aaron},
  year = {2019},
  month = jan,
  journal = {Physical Review A},
  volume = {99},
  number = {1},
  pages = {012323},
  publisher = {{American Physical Society}},
  doi = {10.1103/PhysRevA.99.012323},
  urldate = {2024-01-11}
}

@article{ameriQuantumAlgorithmLinear2023,
  title = {Quantum Algorithm for the Linear {{Vlasov}} Equation with Collisions},
  author = {Ameri, Abtin and Ye, Erika and Cappellaro, Paola and Krovi, Hari and Loureiro, Nuno F.},
  year = {2023},
  month = jun,
  journal = {Physical Review A},
  volume = {107},
  number = {6},
  pages = {062412},
  publisher = {{American Physical Society}},
  doi = {10.1103/PhysRevA.107.062412},
  urldate = {2024-01-11}
}

@article{weckerGatecountEstimatesPerforming2014,
  title = {Gate-Count Estimates for Performing Quantum Chemistry on Small Quantum Computers},
  author = {Wecker, Dave and Bauer, Bela and Clark, Bryan K. and Hastings, Matthew B. and Troyer, Matthias},
  year = {2014},
  month = aug,
  journal = {Physical Review A},
  volume = {90},
  number = {2},
  pages = {022305},
  publisher = {{American Physical Society}},
  doi = {10.1103/PhysRevA.90.022305},
  urldate = {2024-01-11}
}

@article{babbushEncodingElectronicSpectra2018,
  title = {Encoding {{Electronic Spectra}} in {{Quantum Circuits}} with {{Linear T Complexity}}},
  author = {Babbush, Ryan and Gidney, Craig and Berry, Dominic W. and Wiebe, Nathan and McClean, Jarrod and Paler, Alexandru and Fowler, Austin and Neven, Hartmut},
  year = {2018},
  month = oct,
  journal = {Physical Review X},
  volume = {8},
  number = {4},
  pages = {041015},
  publisher = {{American Physical Society}},
  doi = {10.1103/PhysRevX.8.041015},
  urldate = {2023-11-27}
}

@article{haahQuantumAlgorithmSimulating2023,
  title = {Quantum Algorithm for Simulating Real Time Evolution of Lattice {{Hamiltonians}}},
  author = {Haah, Jeongwan and Hastings, Matthew B. and Kothari, Robin and Low, Guang Hao},
  year = {2023},
  month = dec,
  journal = {SIAM Journal on Computing},
  volume = {52},
  number = {6},
  eprint = {1801.03922},
  primaryclass = {quant-ph},
  pages = {FOCS18-250-FOCS18-284},
  issn = {0097-5397, 1095-7111},
  doi = {10.1137/18M1231511},
  urldate = {2023-12-16},
  archiveprefix = {arxiv}
}

@article{mizutaOptimalHamiltonianSimulation2023,
  title = {Optimal {{Hamiltonian}} Simulation for Time-Periodic Systems},
  author = {Mizuta, Kaoru and Fujii, Keisuke},
  year = {2023},
  month = mar,
  journal = {Quantum},
  volume = {7},
  eprint = {2209.05048},
  primaryclass = {cond-mat, physics:quant-ph},
  pages = {962},
  issn = {2521-327X},
  doi = {10.22331/q-2023-03-28-962},
  urldate = {2023-11-27},
  archiveprefix = {arxiv}
}

@article{babbushLowDepthQuantum2018,
  title = {Low {{Depth Quantum Simulation}} of {{Electronic Structure}}},
  author = {Babbush, Ryan and Wiebe, Nathan and McClean, Jarrod and McClain, James and Neven, Hartmut and Chan, Garnet Kin-Lic},
  year = {2018},
  month = mar,
  journal = {Physical Review X},
  volume = {8},
  number = {1},
  eprint = {1706.00023},
  primaryclass = {physics, physics:quant-ph},
  pages = {011044},
  issn = {2160-3308},
  doi = {10.1103/PhysRevX.8.011044},
  urldate = {2024-01-11},
  archiveprefix = {arxiv}
}

@inproceedings{gilyenQuantumSingularValue2019,
  title = {Quantum Singular Value Transformation and beyond: Exponential Improvements for Quantum Matrix Arithmetics},
  shorttitle = {Quantum Singular Value Transformation and Beyond},
  booktitle = {Proceedings of the 51st {{Annual ACM SIGACT Symposium}} on {{Theory}} of {{Computing}}},
  author = {Gily{\'e}n, Andr{\'a}s and Su, Yuan and Low, Guang Hao and Wiebe, Nathan},
  year = {2019},
  month = jun,
  eprint = {1806.01838},
  primaryclass = {quant-ph},
  pages = {193--204},
  doi = {10.1145/3313276.3316366},
  urldate = {2023-07-20},
  archiveprefix = {arxiv}
}

@article{shorPolynomialTimeAlgorithmsPrime1997,
  title = {Polynomial-{{Time Algorithms}} for {{Prime Factorization}} and {{Discrete Logarithms}} on a {{Quantum Computer}}},
  author = {Shor, Peter W.},
  year = {1997},
  month = oct,
  journal = {SIAM Journal on Computing},
  volume = {26},
  number = {5},
  eprint = {quant-ph/9508027},
  pages = {1484--1509},
  issn = {0097-5397, 1095-7111},
  doi = {10.1137/S0097539795293172},
  urldate = {2024-01-11},
  archiveprefix = {arxiv}
}

@misc{hermanSurveyQuantumComputing2022,
  title = {A {{Survey}} of {{Quantum Computing}} for {{Finance}}},
  author = {Herman, Dylan and Googin, Cody and Liu, Xiaoyuan and Galda, Alexey and Safro, Ilya and Sun, Yue and Pistoia, Marco and Alexeev, Yuri},
  year = {2022},
  month = jun,
  number = {arXiv:2201.02773},
  eprint = {2201.02773},
  primaryclass = {quant-ph, q-fin},
  publisher = {{arXiv}},
  urldate = {2022-09-15},
  archiveprefix = {arxiv}
}

@article{stilck2021limitations,
  title={Limitations of optimization algorithms on noisy quantum devices},
  author={Stilck Fran{\c{c}}a, Daniel and Garcia-Patron, Raul},
  journal={Nature Physics},
  volume={17},
  number={11},
  pages={1221--1227},
  year={2021},
  publisher={Nature Publishing Group UK London}
}

@article{zhouWhatLimitsSimulation2020,
  title = {What {{Limits}} the {{Simulation}} of {{Quantum Computers}}?},
  author = {Zhou, Yiqing and Stoudenmire, E. Miles and Waintal, Xavier},
  year = {2020},
  month = nov,
  journal = {Physical Review X},
  volume = {10},
  number = {4},
  pages = {041038},
  publisher = {{American Physical Society}},
  doi = {10.1103/PhysRevX.10.041038},
  urldate = {2024-01-11}
}

@misc{lowQuantumEigenvalueProcessing2024,
  title = {Quantum Eigenvalue Processing},
  author = {Low, Guang Hao and Su, Yuan},
  year = {2024},
  month = jan,
  number = {arXiv:2401.06240},
  eprint = {2401.06240},
  primaryclass = {physics, physics:quant-ph},
  publisher = {{arXiv}},
  doi = {10.48550/arXiv.2401.06240},
  urldate = {2024-01-22},
  archiveprefix = {arxiv}
}

@article{an2023quantum,
  title={Quantum algorithm for linear non-unitary dynamics with near-optimal dependence on all parameters},
  author={An, Dong and Childs, Andrew M and Lin, Lin},
  journal={arXiv preprint arXiv:2312.03916},
  year={2023}
}

@article{reiherElucidatingReactionMechanisms2017,
  title = {Elucidating Reaction Mechanisms on Quantum Computers},
  author = {Reiher, Markus and Wiebe, Nathan and Svore, Krysta M. and Wecker, Dave and Troyer, Matthias},
  year = {2017},
  month = jul,
  journal = {Proceedings of the National Academy of Sciences},
  volume = {114},
  number = {29},
  pages = {7555--7560},
  publisher = {{Proceedings of the National Academy of Sciences}},
  doi = {10.1073/pnas.1619152114},
  urldate = {2024-01-31}
}

@article{sethUniveral1996,
author = {Seth Lloyd },
title = {Universal Quantum Simulators},
journal = {Science},
volume = {273},
number = {5278},
pages = {1073-1078},
year = {1996},
doi = {10.1126/science.273.5278.1073},
URL = {https://www.science.org/doi/abs/10.1126/science.273.5278.1073}}

@article{Childs2019fasterquantum,
  doi = {10.22331/q-2019-09-02-182},
  url = {https://doi.org/10.22331/q-2019-09-02-182},
  title = {Faster quantum simulation by randomization},
  author = {Childs, Andrew M. and Ostrander, Aaron and Su, Yuan},
  journal = {{Quantum}},
  issn = {2521-327X},
  publisher = {{Verein zur F{\"{o}}rderung des Open Access Publizierens in den Quantenwissenschaften}},
  volume = {3},
  pages = {182},
  month = sep,
  year = {2019}
}

@article{ChildsNearlyOpt2019,
  title = {Nearly Optimal Lattice Simulation by Product Formulas},
  author = {Childs, Andrew M. and Su, Yuan},
  journal = {Phys. Rev. Lett.},
  volume = {123},
  issue = {5},
  pages = {050503},
  numpages = {6},
  year = {2019},
  month = {Aug},
  publisher = {American Physical Society},
  doi = {10.1103/PhysRevLett.123.050503},
  url = {https://link.aps.org/doi/10.1103/PhysRevLett.123.050503}
}

@article{ZhaoRandominput2022,
  title = {Hamiltonian Simulation with Random Inputs},
  author = {Zhao, Qi and Zhou, You and Shaw, Alexander F. and Li, Tongyang and Childs, Andrew M.},
  journal = {Phys. Rev. Lett.},
  volume = {129},
  issue = {27},
  pages = {270502},
  numpages = {7},
  year = {2022},
  month = {Dec},
  publisher = {American Physical Society},
  doi = {10.1103/PhysRevLett.129.270502},
  url = {https://link.aps.org/doi/10.1103/PhysRevLett.129.270502}
}

@misc{sunderhauf2023generalized,
      title={Generalized Quantum Singular Value Transformation}, 
      author={Christoph Sünderhauf},
      year={2023},
      eprint={2312.00723},
      archivePrefix={arXiv},
      primaryClass={quant-ph}
}

@article{YuanReviewChem2020,
  title = {Quantum computational chemistry},
  author = {McArdle, Sam and Endo, Suguru and Aspuru-Guzik, Al\'an and Benjamin, Simon C. and Yuan, Xiao},
  journal = {Rev. Mod. Phys.},
  volume = {92},
  issue = {1},
  pages = {015003},
  numpages = {51},
  year = {2020},
  month = {Mar},
  publisher = {American Physical Society},
  doi = {10.1103/RevModPhys.92.015003},
  url = {https://link.aps.org/doi/10.1103/RevModPhys.92.015003}
}

@article{QchemSimu2001,
  title = {Quantum algorithms for fermionic simulations},
  author = {Ortiz, G. and Gubernatis, J. E. and Knill, E. and Laflamme, R.},
  journal = {Phys. Rev. A},
  volume = {64},
  issue = {2},
  pages = {022319},
  numpages = {14},
  year = {2001},
  month = {Jul},
  publisher = {American Physical Society},
  doi = {10.1103/PhysRevA.64.022319},
  url = {https://link.aps.org/doi/10.1103/PhysRevA.64.022319}
}

@article{ThermalRateChem1999,
  title = {Calculating the thermal rate constant with exponential speedup on a quantum computer},
  author = {Lidar, Daniel A. and Wang, Haobin},
  journal = {Phys. Rev. E},
  volume = {59},
  issue = {2},
  pages = {2429--2438},
  numpages = {0},
  year = {1999},
  month = {Feb},
  publisher = {American Physical Society},
  doi = {10.1103/PhysRevE.59.2429},
  url = {https://link.aps.org/doi/10.1103/PhysRevE.59.2429}
}

@article{anLinearCombinationHamiltonian2023,
  title = {Linear {{Combination}} of {{Hamiltonian Simulation}} for {{Nonunitary Dynamics}} with {{Optimal State Preparation Cost}}},
  author = {An, Dong and Liu, Jin-Peng and Lin, Lin},
  year = {2023},
  month = oct,
  journal = {Physical Review Letters},
  volume = {131},
  number = {15},
  pages = {150603},
  publisher = {{American Physical Society}},
  doi = {10.1103/PhysRevLett.131.150603},
  urldate = {2023-10-22}
}

@phdthesis{KothariRobin2014, 
author={{Kothari, Robin}}, 
school = {University of Waterloo},
title={Efficient algorithms in quantum query complexity}, 
year={2014}, 
publisher="UWSpace", 
url={http://hdl.handle.net/10012/8625} 
}

@misc{lowHamiltonianSimulationInteraction2019,
  title = {Hamiltonian {{Simulation}} in the {{Interaction Picture}}},
  author = {Low, Guang Hao and Wiebe, Nathan},
  year = {2019},
  month = jun,
  number = {arXiv:1805.00675},
  eprint = {1805.00675},
  primaryclass = {quant-ph},
  publisher = {arXiv},
  doi = {10.48550/arXiv.1805.00675},
  urldate = {2024-03-15},
  archiveprefix = {arXiv}
}

@article{an_quantum_2022,
	title = {Quantum linear system solver based on time-optimal adiabatic quantum computing and quantum approximate optimization algorithm},
	volume = {3},
	issn = {2643-6809, 2643-6817},
	url = {http://arxiv.org/abs/1909.05500},
	doi = {10.1145/3498331},
	number = {2},
	urldate = {2022-11-26},
	journal = {ACM Transactions on Quantum Computing},
	author = {An, Dong and Lin, Lin},
	month = jun,
	year = {2022},
	note = {Number: 2
arXiv:1909.05500 [quant-ph]},
	pages = {1--28}
}

@article{liuEfficientQuantumAlgorithm2021,
  title = {Efficient Quantum Algorithm for Dissipative Nonlinear Differential Equations},
  author = {Liu, Jin-Peng and Kolden, Herman {\O}ie and Krovi, Hari K. and Loureiro, Nuno F. and Trivisa, Konstantina and Childs, Andrew M.},
  year = {2021},
  month = aug,
  journal = {Proceedings of the National Academy of Sciences},
  volume = {118},
  number = {35},
  pages = {e2026805118},
  publisher = {Proceedings of the National Academy of Sciences},
  doi = {10.1073/pnas.2026805118},
  urldate = {2024-03-26}
}

@misc{anTheoryQuantumDifferential2023,
  title = {A Theory of Quantum Differential Equation Solvers: Limitations and Fast-Forwarding},
  shorttitle = {A Theory of Quantum Differential Equation Solvers},
  author = {An, Dong and Liu, Jin-Peng and Wang, Daochen and Zhao, Qi},
  year = {2023},
  month = mar,
  number = {arXiv:2211.05246},
  eprint = {2211.05246},
  primaryclass = {quant-ph},
  publisher = {arXiv},
  doi = {10.48550/arXiv.2211.05246},
  urldate = {2024-08-17},
  archiveprefix = {arXiv}
}

@article{kroviImprovedQuantumAlgorithms2023,
	title = {Improved quantum algorithms for linear and nonlinear differential equations},
	volume = {7},
	issn = {2521-327X},
	url = {http://arxiv.org/abs/2202.01054},
	doi = {10.22331/q-2023-02-02-913},
	urldate = {2023-09-30},
	journal = {Quantum},
	author = {Krovi, Hari},
	month = feb,
	year = {2023},
	note = {arXiv:2202.01054 [physics, physics:quant-ph]},
	pages = {913}
}

@article{fangTimemarchingBasedQuantum2023,
  title = {Time-Marching Based Quantum Solvers for Time-Dependent Linear Differential Equations},
  author = {Fang, Di and Lin, Lin and Tong, Yu},
  year = {2023},
  month = mar,
  journal = {Quantum},
  volume = {7},
  eprint = {2208.06941},
  primaryclass = {quant-ph},
  pages = {955},
  issn = {2521-327X},
  doi = {10.22331/q-2023-03-20-955},
  urldate = {2024-08-17},
  archiveprefix = {arXiv}
}

@article{clinton2024towards,
  title={Towards near-term quantum simulation of materials},
  author={Clinton, Laura and Cubitt, Toby and Flynn, Brian and Gambetta, Filippo Maria and Klassen, Joel and Montanaro, Ashley and Piddock, Stephen and Santos, Raul A and Sheridan, Evan},
  journal={Nature Communications},
  volume={15},
  number={1},
  pages={211},
  year={2024},
  publisher={Nature Publishing Group UK London}
}

@article{campbel2019random,
  title = {Random Compiler for Fast Hamiltonian Simulation},
  author = {Campbell, Earl},
  journal = {Phys. Rev. Lett.},
  volume = {123},
  issue = {7},
  pages = {070503},
  numpages = {5},
  year = {2019},
  month = {Aug},
  publisher = {American Physical Society},
  doi = {10.1103/PhysRevLett.123.070503},
  url = {https://link.aps.org/doi/10.1103/PhysRevLett.123.070503}
}

@article{wan2022random,
  title = {Randomized Quantum Algorithm for Statistical Phase Estimation},
  author = {Wan, Kianna and Berta, Mario and Campbell, Earl T.},
  journal = {Phys. Rev. Lett.},
  volume = {129},
  issue = {3},
  pages = {030503},
  numpages = {7},
  year = {2022},
  month = {Jul},
  publisher = {American Physical Society},
  doi = {10.1103/PhysRevLett.129.030503},
  url = {https://link.aps.org/doi/10.1103/PhysRevLett.129.030503}
}

@article{yang2021accelerate,
  title = {Accelerated Quantum Monte Carlo with Mitigated Error on Noisy Quantum Computer},
  author = {Yang, Yongdan and Lu, Bing-Nan and Li, Ying},
  journal = {PRX Quantum},
  volume = {2},
  issue = {4},
  pages = {040361},
  numpages = {23},
  year = {2021},
  month = {Dec},
  publisher = {American Physical Society},
  doi = {10.1103/PRXQuantum.2.040361},
  url = {https://link.aps.org/doi/10.1103/PRXQuantum.2.040361}
}

@misc{zeng2022simplehighprecisionhamiltoniansimulation,
      title={Simple and high-precision Hamiltonian simulation by compensating Trotter error with linear combination of unitary operations}, 
      author={Pei Zeng and Jinzhao Sun and Liang Jiang and Qi Zhao},
      year={2022},
      eprint={2212.04566},
      archivePrefix={arXiv},
      primaryClass={quant-ph},
      url={https://arxiv.org/abs/2212.04566}, 
}

@article{granetHamiltonianDynamicsDigital2024,
  title = {Hamiltonian Dynamics on Digital Quantum Computers without Discretization Error},
  author = {Granet, Etienne and Dreyer, Henrik},
  year = {2024},
  month = sep,
  journal = {npj Quantum Information},
  volume = {10},
  number = {1},
  pages = {1--10},
  publisher = {Nature Publishing Group},
  issn = {2056-6387},
  doi = {10.1038/s41534-024-00877-y},
  urldate = {2024-09-27},
  copyright = {2024 The Author(s)}
}

@article{Hsieh2024double,
  title = {Doubling the order of approximation via the randomized product formula},
  author = {Cho, Chien-Hung and Berry, Dominic W. and Hsieh, Min-Hsiu},
  journal = {Phys. Rev. A},
  volume = {109},
  issue = {6},
  pages = {062431},
  numpages = {11},
  year = {2024},
  month = {Jun},
  publisher = {American Physical Society},
  doi = {10.1103/PhysRevA.109.062431},
  url = {https://link.aps.org/doi/10.1103/PhysRevA.109.062431}
}

@misc{jiangOptimalSpaceDepthTradeCNOT2022,
  title = {Optimal {{Space-Depth Trade-Off}} of {{CNOT Circuits}} in {{Quantum Logic Synthesis}}},
  author = {Jiang, Jiaqing and Sun, Xiaoming and Teng, Shang-Hua and Wu, Bujiao and Wu, Kewen and Zhang, Jialin},
  year = {2022},
  month = sep,
  number = {arXiv:1907.05087},
  eprint = {1907.05087},
  primaryclass = {quant-ph},
  publisher = {arXiv},
  urldate = {2024-10-02},
  archiveprefix = {arXiv}
}

@article{martyn2025halving,
  title={Halving the cost of quantum algorithms with randomization},
  author={Martyn, John M and Rall, Patrick},
  journal={npj Quantum Information},
  volume={11},
  number={1},
  pages={47},
  year={2025},
  publisher={Nature Publishing Group UK London}
}

\appendix

\onecolumngrid

\section{Framework implementation on BCCKS algorithm}
\label{app:LCU}

\begin{lemma}{(Ref.~\cite{berrySimulatingHamiltonianDynamics2015}. Error and success probability of $V_1$)}
\label{le:lcu_v1}

    The quantum circuit $V_1$ implements $\U$, approximating the unitary $U = e^{-iH\tau}$ with error bounded by 
\begin{equation}
\begin{aligned}
    &\|V_1-U\|\le a_1 \\
    &a_1 =\delta_1(\frac{\delta^2_1+3\delta_1+4}{2}),
\end{aligned}
\end{equation}
where $\delta_1=2\frac{(\ln 2)^{K_1+1}}{(K_1+1)!}$. The success probability is lower bounded by $\theta_1 = (1-a_1)^2$.
\end{lemma}
\begin{proof}

We can bound the truncation error of $\U$ by
    \begin{equation}
        \begin{aligned}
            \left\|\U - U\right\| &= \left\|\sum_{k=K_1+1}^\infty \frac{\left(-iH\tau\right)^k}{k!}\right\|\\
            & \le \sum_{k=K_1+1}^\infty \frac{\left(\|H\|t\right)^k}{k!} \\
            & \le \sum_{k=K_1+1}^\infty \frac{\left(\tau\sum^L_{l=1}\alpha_l\right)^k}{k!} \\
            &= \sum_{k=K_1+1}^\infty \frac{\left(\ln 2\right)^k}{k!} \\
            & \le 2\frac{\left(\ln 2\right)^{K_1+1}}{(K_1+1)!} =: \delta_1.\\
        \end{aligned}
    \end{equation}
    Therefore, the following holds,
    \begin{equation}
        \left\|\U\right\| \le \left\|\U-U\right\| + \|U\| \le 1 + \delta_1
    \end{equation}
    and
    \begin{equation}
        \begin{aligned}
            \left\|\U\UD - \mathds{1}\right\| &\le \left\|\U\UD - U\UD\right\| + \left\|U\UD - UU^\dagger\right\|  \\
            &\le \delta_1\left(1 + \delta_1\right) + \delta_1 = \delta_1\left(2+\delta_1\right).
        \end{aligned}
    \end{equation}
    Eventually, we complete the proof of error bound by 
    \begin{equation}
        \begin{aligned}
            \left\|V_1 - U\right\| &= \left\|\frac{3}{2}\U - \frac{1}{2}\U\UD \U - U\right\| \\
            &\le \left\|\U - U\right\| + \frac{1}{2}\left\|\U - \U\UD \U\right\| \\ 
            &\le \delta_1 + \frac{\delta_1(1+\delta_1)(2+\delta_1)}{2} = \delta_1(\frac{\delta^2_1+3\delta_1+4}{2}) = a_1.
        \end{aligned}
    \end{equation}
    As for the success probability, we apply lemma G.4. in ref. \cite{childsFirstQuantumSimulation2018} to claim that it is greater than $\left(1-a_1\right)^2$
\end{proof}

\begin{lemma}{(Error and success probability of $V_2$)}
\label{le:lcu_v2}
There exists a quantum circuit $V_2$ implementing $\UU$, and $V_2$ approximates the unitary $U = e^{-iH\tau}$ with error bounded by 
\begin{equation}
    \begin{aligned}
        \|V_2-U\|\le \left(1+ \frac{40}{s_2^3} + \frac{64}{s_2^5}\right)\delta_2 =: a_2
    \end{aligned}
\end{equation}
where $\delta_2 =\frac{p}{1-p}\frac{(\ln 2)^{K_1+1}}{(K_1+1)!}$ and $s_2 = (\sin(\pi/10))^{-1}$. The success probability is lower bounded by  $\theta_2 = (1-a_2)^2$.
\end{lemma}
\begin{proof}

    The operator we are approximating using $V_2$ is 
    \begin{equation}
    \label{eq:F_V_2}
        \begin{aligned}
        \UU&=\sum_{k=0}^{K_1}\frac{(-iHt)^k}{k!}+\frac{1}{1-p}\sum_{k=K_1+1}^{K_2}\frac{(-iHt)^k}{k!}.\\
        \end{aligned}
    \end{equation}
    With another set of oracles as Eq.~\eqref{eq:lcu_oracle}, $G_2$ preparing the coefficients and $\mathbf{\rm SELECT}(V_2)$ applying unitaries for Eq.~\eqref{eq:F_V_2}, we can construct $W_2 = (G_2^{\dag}\otimes \mathds{1}) \mathbf{\rm SELECT}(\UU) (G_2\otimes \mathds{1}),$ such that
    \begin{equation}
        W_2(\ket{0}\otimes \ket{\psi}) = \frac{1}{s^\prime_2}(\ket{0}\otimes  \UU\ket{\psi})+\ket{\perp},
    \end{equation}
    where $(\bra{0}\otimes \mathds{1})\ket{\perp} = 0$ and
    \begin{equation}
        s_2^\prime = \sum_{k=0}^{K_1}\frac{\left(\tau\sum_{l=1}^L\alpha_l\right)^k}{k!}+  \frac{1}{1-p}\sum_{k=K_1+1}^{K_2}\frac{\left(\tau\sum_{l=1}^L\alpha_l\right)^k}{k!}\approx \exp(\ln 2)+\frac{p (\ln 2)^{K_1+1}}{(1-p)(K_1+1)!}\exp(\ln 2)
        =2+\frac{2p}{(1-p)}\delta_1.
    \end{equation}
    We want to amplify $1/s_2^\prime$ to 1 by applying OAA. Although $s_2^\prime$ is unbounded above when $p \to 1$, $\{K_1,K_2,p\}$ with extreme $p$ will be discarded when transversing viable sets for a given cost budget. Therefore, it is safe for us to bound $p \le 1/(1+2\delta_1)$ such that $s_2^\prime \le 3$ and further amplify $s_2^\prime$ to $s_2 = (\sin(\pi/10))^{-1}$.

    We thus perform $PTTW$ such that
    \begin{equation}
        \begin{aligned}
            PTTW \left(\ket{0}\otimes \ket{\psi}\right)=\ket{0}\otimes\left(\frac{5}{s_2}\UU-\frac{20}{s^3_2}\UU\UUD \UU+\frac{16}{s^5_2}\UU \UUD \UU \UUD \UU\right)\ket{\psi}.
        \end{aligned}
    \end{equation}
    Finally, we obtain the near-unitary operator
    \begin{equation}
    \label{eq:sup_LCU_V2}
        \begin{aligned}
            V_2 &= \left(\bra{0}\otimes \mathds{1}\right) PTTW \left(\ket{0}\otimes \mathds{1}\right) \\
            &=\left(\frac{5}{s_2}-\frac{20}{s^3_2}+\frac{16}{s^5_2}\right) \UU-\frac{20}{s^3_2}\left(\UU\UUD \UU-\UU\right)+ \frac{16}{s^5_2} \left(\UU \UUD \UU \UUD \UU-\UU\right).
        \end{aligned}
    \end{equation}
    We denote the truncated error in $\UU$ as $\delta_2$, where 
    \begin{equation}
        \delta_2 =\left|-\frac{p}{1-p}\frac{(\ln 2)^{K_1+1}}{(K_1+1)!}+ \frac{2}{1-p} \frac{(\ln 2)^{K_2+1}}{(K_2+1)!}\right| \le  \frac{p}{1-p}\frac{(\ln 2)^{K_1+1}}{(K_1+1)!} = \frac{p}{1-p}\delta_1
    \end{equation}
    With the facts $\left\|\UU\right\|\le 1+\delta_2$, $\left\|\UUD \UU-I\right\|\le \delta_2(2+\delta_2)$, and
    \begin{equation}
    \begin{aligned}
         \left\|\UUD \UU \UUD \UU-\mathds{1}\right\|&\le \|\UUD \UU \UUD \UU-\UUD \UU \|+ \|\UUD \UU-\mathds{1}\|\\
         &\le \left\|\UUD \UU\right\| \delta_2(2+\delta_2)+ \delta_2(2+\delta_2)\\
         &\le [(1+\delta_2)^2+1]\delta_2(2+\delta_2)=\delta_2(2+\delta_2)(2+2\delta_2+\delta_2^2)\\
         &\approx 4\delta_2.
    \end{aligned}
    \end{equation}
    We can thereby compute $\|V_2\|$ terms by terms. In the first term in the last line of Eq.~\eqref{eq:sup_LCU_V2}, we have 
    \begin{equation}
        \left\|\left(\frac{5}{s_2}-\frac{20}{s^3_2}+\frac{16}{s^5_2}\right) \UU - U \right\|\le \delta, 
    \end{equation}
    as $\frac{5}{s_2}-\frac{20}{s^3_2}+\frac{16}{s^5_2} \le 1$. 
    For the rest,
    \begin{equation}
    \begin{aligned}
    \left\|\frac{20}{s^3_2}\left(\UU\UUD \UU-\UU\right)\right\| &\le \frac{20}{s^3_2}(1+\delta_2)\delta_2(2+\delta_2) \le \frac{80}{s^3_2}\delta_2\\
    \left\|\frac{16}{s^5_2}  \left(\UU \UUD \UU \UUD \UU-\UU\right)\right\|&\le \frac{16}{s^5_2} (1+\delta_2)\delta_2(2+\delta_2)(2+2\delta_2+\delta_2^2) \le \frac{128}{s^5_2} \delta_2.
    \end{aligned}
    \end{equation}
    Consequently, 
    \begin{equation}
        \|V_2-U\|\le \left(1+\frac{80}{s^3_2}+ \frac{128}{s^5_2} \right)\delta_2 \le 4\delta_2 = a_2.
    \end{equation}
    The success probability follows similarly to be greater than $(1-a_2)^2$
\end{proof}

The last ingredient is the calculation of the error bound on the operator $V_{m} = pV_1 + (1-p) V_2$. Since we never actually implement $V_{\rm m}$ we do not need its cost. 

\begin{lemma}{(Error of $V_{m}$)}
\label{le:lcu_vm}
The quantum circuit $V_m$ implemented by mixing two quantum circuits $V_1$ and $V_2$ with probability $p$ and $1-p$ respectively approximates $U = e^{-iH\tau}$ with bounded error
 \begin{equation}
 \label{eq:lcu_vm}
    \begin{aligned}
        &\|V_m - U\|\le b \\
        &b =\delta_m + \frac{3}{1-p}\delta_1^2,
    \end{aligned}
 \end{equation}
 where $\delta_m=2\frac{(\ln 2)^{K_2+1}}{(K_2+1)!}$.
 \end{lemma}
\begin{proof}

    We cannot trivially add each term in $\U$ and $\UU$ linearly because each of $V_1$ and $V_2$ is reflected during OAA. We, therefore, give a loose upper bound by analyzing each term in the error sources separately. 

    Observe that we have $U = \U + E_1$ and $U = \UU + E_2$, where $E_1$ and $E_2$ are truncation errors with, 
    \begin{equation}
        \begin{aligned}
            E_1&=\sum_{k=K_1+1}^{\infty}(-iH\tau)^k,\\
            E_2&= -\frac{p}{1-p}\sum_{k=K_1+1}^{K_2}(-iH\tau)+\sum_{k=K_2}^{\infty}(-iH\tau)^k,\\
        \end{aligned}
    \end{equation} 
    and we can bound their absolute value by $\delta_1$ and $\delta_2$ respectively. We can thus express
    \begin{equation}
    \label{eq:v12}
        \begin{aligned}
            V_1&=U+\frac{1}{2}(E_1-U^\dag E_1U)+ R_1\\
            V_2&=U+ \frac{1}{2}(E_2-U^\dag E_2U)+R_2,
        \end{aligned}
    \end{equation}
    where $\|R_1\|\le \frac{3}{2}\|E_1\|^2+\frac{1}{2}\|E_1\|^3\le \frac{3}{2}\delta_1^2+\frac{1}{2} \delta_1^3$ and $\|R_2\|\le (\frac{80}{s^3}+\frac{128}{s^5} ) \|E_2\|^2 +\mathcal O(\|E_2\|^3)\le \frac{3p^2}{(1-p)^2}\delta_1^2$ are the truncation error after OAA. These bounds can be derived after invoking lemma \ref{le:lcu_v1} and \ref{le:lcu_v2}.
    
    Lastly, note that 
    \begin{equation}
    \label{eq:PE12}
        \left\|pE_1+(1-p)E_2\right\|= \left\|\sum_{k=K_2}^{\infty}(-iH\tau)^k\right\|\le 2\frac{(\ln 2)^{K_2+1}}{(K_2+1)!} =: \delta_m,
    \end{equation}
    combining Eq.~\eqref{eq:v12} and \eqref{eq:PE12}, we have 
    \begin{equation}
        \begin{aligned}
            \left\|pV_1+(1-p)V_2-U\right\|&=\left\|\frac{1}{2} [(pE_1+(1-p)E_2)-  U_0^\dag (pE_1+(1-p)E_2)U_0 ]+pR_1+(1-p)R_2\right\|\\
            &\le \delta_m+ p\left(\frac{3}{2}\delta_1^2+\frac{1}{2} \delta_1^3\right) + (1-p)\left(\frac{3p^2}{(1-p)^2}\delta_1^2\right) \\
            &\le  \delta_m + \frac{3}{1-p}\delta_1^2 + \mathcal{O}\left(\delta_1^3\right) =: b,
        \end{aligned}
    \end{equation}
    where the inequality in second line holds because an operator $O$ satisfy $\|O\| \le \|U^\dagger O U\|$ for any unitary $U$
\end{proof}

\subsection{Proof of Corollary~\ref{cor:LCU}}
\begin{proof}
    From lemma \ref{le:mix}, we know the error after mixing channel can be expressed by $a_1$ $a_2$ and $b$, which were derived in lemma \ref{le:lcu_v1}, \ref{le:lcu_v2} and \ref{le:lcu_vm}. We can combine the results to get
    \begin{equation}
    \label{eq:app_lcu_final_error}
    \begin{aligned}
        \epsilon &= 4b + 2pa_1^2 + 2(1-p)a_2^2 \\
        &= 4\left( \delta_m + \frac{3}{1-p}\delta_1^2 + \mathcal{O}\left(\delta_1^3\right)\right) + 2(1-p)\left(1+\frac{80}{s^3_2}+ \frac{128}{s^5_2} \right)^2\delta_2^2 + 2p\left(\delta_1(\frac{\delta^2_1+3\delta_1+4}{2})\right)^2\\
        &\le \frac{20}{1-p}\delta_1^2 +4\delta_m \\
        &\le \max\left\{\frac{40}{1-p}\delta_1^2, 8\delta_m\right\},
    \end{aligned}
    \end{equation}
    where the equality in the last line holds because $\delta_m$ is exponentially smaller than $\delta_1$.
    As for the success probability, according to  \cite{childsFirstQuantumSimulation2018}, the lower bound on success probability for each of implementing $V_1$ and $V_2$ are $(1-a_1)^2$ and $(1-a_2)^2$ respectively. Therefore, the overall algorithm succeeds with the probability of at least
    \begin{equation}
        \begin{aligned}
            \theta &\ge p(1-a_1)^2 + (1-p)(1-a_2)^2 \\
            &\ge  1- \frac{8}{1-p}\delta^2_1 + \delta_1^2 - 4\delta_1.
        \end{aligned}
    \end{equation}
    It implies that the failure probability $\xi \le \frac{8}{1-p}\delta^2_1 + 4\delta_1$.
\end{proof}

\section{QSP}
\label{app:QSP}

\begin{lemma}{(Truncation and rescaling error)}

For index $i = \{1,2\}$, we can upper bound 

\begin{equation}
    \left\|\bra{+}_b\bra{G}_a\hat{V}_i\ket{G}_a\ket{+}_b - e^{-iHt}\right\|
\end{equation}

by $\mathcal{O}\left(\sqrt{|\epsilon_{1,V_i}|+\epsilon_{2,V_i}}\right)$ where $\epsilon_{1,V_i} := \tilde{A}_i(0) - 1$ and $\epsilon_{2,V_i} := \displaystyle\max_{\lambda} \tilde{A}^2(\lambda) + \tilde{C}^2(\lambda) - 1$.
\label{le:qsp_result}
\begin{enumerate}
    \item 
    \begin{equation}
        \epsilon_{1,V_1} \le \frac{4t^{K_1}}{2^{K_1} K_1!}, \quad \epsilon_{2,V_1} = 0
    \end{equation}
    \item 
    \begin{equation}
        \epsilon_{1,V_2} \le \frac{p}{1-p}\frac{4t^{K_1}}{2^{K_1} K_1!}, \quad \epsilon_{2,V_2} = \frac{5p}{1-p}\frac{4t^{K_1}}{2^{K_1} K_1!}
    \end{equation}
    \item For the mixed function, $V_m = pV_1 + (1-p)V_2$, we have 
    \begin{equation}
        \epsilon_{1,V_m} \le \frac{4t^{K_2}}{2^{K_2} K_2!}, \quad \epsilon_{2,V_m} = 0
    \end{equation}
\end{enumerate}
\end{lemma}

\begin{proof}

    \begin{enumerate}
        \item proof for $\hat{V}_1$

        $\epsilon_{1,V_1}$ is the truncation error in Eq.~\eqref{eq:qsp_ac} with $K = K_1$. Since $\lambda$ originated from a cosine function in the Chebyshev polynomial, it takes the value $\lambda\in[-1,1]$, and we maximize the truncation error over the domain to obtain an upper bound on $\epsilon_{1,V_1}$. Thus, 
        \begin{equation} 
            \begin{aligned}
                \epsilon_{1,V_1}&\le\max_{\lambda\in[-1,1]}\left| e^{-i\lambda t} - \left(\tilde{A}_1(\lambda) + i\tilde{C}_1(\lambda)\right)\right| \\
                &=\max_{\lambda\in[-1,1]}\left|2\sum^{K_1}_{\text{even } k=1}(-1)^{\frac{k}{2}}J_k(t)T_k(\lambda) + i2\sum^{K_1}_{\text{odd } k=1}(-1)^{\frac{k-1}{2}}J_k(t)T_k(\lambda)\right|\\
                & \le 2\sum^{K_1}_{k=1}|J_k(t)| \le \frac{4t^{K_1}}{2^{K_1} K_1!}.
            \end{aligned}
        \end{equation}

        It is trivial that $\tilde{A}^2_1(\lambda)+\tilde{C}^2_1(\lambda) \le A^2(\lambda)+ C^2(\lambda) = 1, \forall K_1 \in \mathbb{Z}$. Thus $\epsilon_{2,V_1} = 0$
        \item Proof for $\hat{V}_2$

        Similarly, we obtain 
        \begin{equation}
        \label{eq:e1v2}
            \begin{aligned}
                \epsilon_{1,V_2}& \le\max_{\lambda\in[-1,1]}\left| e^{-i\lambda t} - \left(\tilde{A}_2(\lambda) + i\tilde{C}_2(\lambda)\right)\right| \\
                &=\max_{\lambda\in[-1,1]}\left|
                -\frac{p}{1-p}\left(2\sum^{K_2}_{\text{even} k>K_1}(-1)^{\frac{k}{2}}J_k(t)T_k(\lambda) + i2\sum^{K_2}_{\text{odd} k>K_1}(-1)^{\frac{k-1}{2}}J_k(t)T_k(\lambda)\right)\right. \\
                &\left. \quad \qquad \qquad +2\sum^{\infty}_{\text{even} k> K_2}(-1)^{\frac{k}{2}}J_k(t)T_k(\lambda) + i2\sum^{\infty}_{\text{odd} k>K_2}(-1)^{\frac{k-1}{2}}J_k(t)T_k(\lambda)\right| \\
                &\le \frac{p}{1-p}\left(2\sum^{K_2}_{k>K_1}|J_k(t)| \right)-2\sum^{\infty}_{k> K_2}|J_k(t)| \\
                &=\frac{p}{1-p}\left(2\sum^{\infty}_{k>K_1}|J_k(t)| \right) -\left(2 +\frac{p}{1-p}\right)\sum^{\infty}_{k> K_2}|J_k(t) |\\
                &\le \frac{p}{1-p}\frac{4t^{K_1}}{2^{K_1} K_1!}.
            \end{aligned}
        \end{equation}
        It should be noticed that $\epsilon_{2,V_2}$ will be greater than zero as we increase $p$ for a given $K_2$. We have to find $\epsilon_{2,V_2}$ satisfying 
        \begin{equation}
        \label{eq:e2v2}
            \begin{aligned}
                \epsilon_{2,V_2} &\le \left|1- \left(\tilde{A}^2_2(\lambda) + \tilde{C}^2_2(\lambda)\right)\right| \\
                & = \left|\left(A^2(\lambda)-\tilde{A}^2_2(\lambda)\right) + \left(C^2(\lambda) - \tilde{C}^2_2(\lambda)\right)\right|.
            \end{aligned}
        \end{equation}
        For simplicity, we abbreviate the sum by $2\sum^n_{\text{even } k>m}(-1)^{\frac{k}{2}}J_k(t)T_k(\lambda) = \mathfrak{S}_m^n(\lambda)$ and compute the first parentheses in Eq.~\eqref{eq:e2v2} as 
        \begin{equation}
            \begin{aligned}
                &A^2(\lambda)-\tilde{A}^2_2(\lambda) \\
                &=\left(J_0(t) + \mathfrak{S}_0^\infty(\lambda) \right)^2 -\left(J_0(t) + \mathfrak{S}_0^{K_1}(\lambda) + \frac{1}{1-p}\left(\mathfrak{S}_{K_1}^{K_2}(\lambda)\right) \right)^2 \\
                &= 2J_0(t)\mathfrak{S}_0^\infty(\lambda) + \left(\mathfrak{S}_0^\infty(\lambda)\right)^2 -2J_0(t)\left(\mathfrak{S}_0^{K_1}(\lambda) + \frac{1}{1-p}\left(\mathfrak{S}_{K_1}^{K_2}(\lambda)\right)\right) - \left(\mathfrak{S}_0^{K_1}(\lambda) + \frac{1}{1-p}\left(\mathfrak{S}_{K_1}^{K_2}(\lambda)\right) \right)^2 \\
                &= 2J_0\left(\mathfrak{S}_0^\infty(\lambda) -\mathfrak{S}_0^{K_1}(\lambda) - \frac{1}{1-p}\left(\mathfrak{S}_{K_1}^{K_2}(\lambda)\right)\right) + \left(\mathfrak{S}_0^\infty(\lambda)\right)^2  - \left(\mathfrak{S}_0^{K_1}(\lambda) + \frac{1}{1-p}\left(\mathfrak{S}_{K_1}^{K_2}(\lambda)\right) \right)^2 \\
                &= 2J_0\delta_s + \left(\mathfrak{S}_0^\infty(\lambda)+\mathfrak{S}_0^{K_1}(\lambda) + \frac{1}{1-p}\left(\mathfrak{S}_{K_1}^{K_2}(\lambda)\right)\right)\left(\mathfrak{S}_0^\infty(\lambda) -\mathfrak{S}_0^{K_1}(\lambda) - \frac{1}{1-p}\left(\mathfrak{S}_{K_1}^{K_2}(\lambda)\right)\right) \\
                &\le (2J_0+3)\delta_s,
            \end{aligned}
        \end{equation}
where 
\begin{equation}
\label{eq:qsp_hs_delta_s}
\begin{aligned}
    |\delta_s| &= \left|\mathfrak{S}_0^\infty(\lambda) -\mathfrak{S}_0^{K_1}(\lambda) - \frac{1}{1-p}\left(\mathfrak{S}_{K_1}^{K_2}(\lambda)\right)\right| \\
    &= \left|-\frac{p}{1-p}\mathfrak{S}_{K_1}^{\infty}(\lambda) + \frac{1}{1-p}\mathfrak{S}_{\ket{_2}}^{\infty}(\lambda)\right|\\
    &\le \frac{p}{1-p}\mathfrak{S}_{K_1}^{\infty}(\lambda).
\end{aligned}
\end{equation}

        Similarly, with definition $2\sum^n_{\text{odd } k>m}(-1)^{\frac{k-1}{2}}J_k(t)T_k(\lambda) = \mathfrak{K}_m^n(\lambda)$ we have 
        \begin{equation}
            \begin{aligned}
                &C^2(\lambda)-\tilde{C}^2_2(\lambda) \approx 3\delta_k,
            \end{aligned}
        \end{equation}
        where 
        \begin{equation}
        \label{eq:qsp_hs_delta_k}
        \begin{aligned}
            |\delta_k| &= \left|\mathfrak{K}_0^\infty(\lambda) -\mathfrak{K}_0^{K_1}(\lambda) - \frac{1}{1-p}\left(\mathfrak{K}_{K_1}^{K_2}(\lambda)\right)\right| \\
            &\le \frac{p}{1-p}\mathfrak{K}_{K_1}^{\infty}(\lambda).
        \end{aligned}
        \end{equation}
        Observe that 
        \begin{equation}
            \delta_s + \delta_k \le \frac{p}{1-p}\frac{4t^{K_1}}{2^{K_1} K_1!}
        \end{equation} and $J_0(t) \le 1, \forall t$, we finally bound  
        \begin{equation}
            \epsilon_{2,V_2} \le (2J_0+3)\delta_s + 3\delta_k \le \frac{5p}{1-p}\frac{4t^{K_1}}{2^{K_1} K_1!}.
        \end{equation}

        \item Proof for $\hat{V}_m$

        Observe that $V_m = pV_1 + (1-p) V_2 = J_0(t) + 2\sum^{K_2}_{{\rm even}\, k>0}(-1)^{k/2}J_k(t)T_k(\lambda) + i2\sum^{K_2}_{{\rm odd}\,  k>0}(-1)^{(k-1)/2}J_k(t)T_k(\lambda)$, we have 
        \begin{equation}
            \epsilon_{1,V_m} \le \frac{4t^{K_2}}{2^{K_2} K_2!} \text{ and } \epsilon_{2,V_m} = 0
        \end{equation}
    \end{enumerate}
\end{proof}

We can now prove the error bound on HS with QSP with lemma \ref{le:qsp_result} and \ref{le:qsp_functions_conditions}

\subsection{Proof of Corollary~\ref{cor:QSP_HS}}
\begin{proof}

The error of mixing channel in lemma \ref{le:mix}. 

\begin{equation}
    \begin{aligned}
        \epsilon &= 4b + 2pa_1^2 + 2(1-p)a_2^2 \\
        &= 4 \sqrt{\frac{4t^{K_2}}{2^{K_2} K_2!}} + 2p \frac{4t^{K_1}}{2^{K_1} K_1!} + 2(1-p)\frac{6p}{1-p}\frac{4t^{K_1}}{2^{K_1} K_1!} \\
        &\le 4\sqrt{\frac{4t^{K_2}}{2^{K_2} K_2!}} + 14\frac{4t^{K_1}}{2^{K_1} K_1!} \\
        &\le \max\left\{8\sqrt{\frac{4t^{K_2}}{2^{K_2} K_2!}},28\frac{4t^{K_1}}{2^{K_1} K_1!}\right\}.
    \end{aligned}
\end{equation}

We can lower bound the failure probability $\xi$ using the lemma \ref{le:mix} and \ref{le:qsp_functions_conditions} such that

\begin{equation}
\label{eq:QSP_HS_xi}
\begin{aligned}
    \xi &\le 2pa_1 + 2(1-p)a_2 \\
    &= 2p \sqrt{\frac{4t^{K_1}}{2^{K_1} K_1!}} + 2 \sqrt{6p(1-p)}\sqrt{\frac{4t^{K_1}}{2^{K_1} K_1!}} \\
    &\le 4p \sqrt{\frac{4t^{K_1}}{2^{K_1} K_1!}}.
\end{aligned}
\end{equation}
\end{proof}

\subsection{Uniform spectral amplification (USA)}

\begin{lemma}{(Truncation error of approximating error function)}
\label{le:QSP_erf_error}

    Polynomial functions, $ P_{erf,\gamma ,K_1}(\lambda)$ and $ P_{erf,\gamma ,K_2}(\lambda)$, constructed by Jacobi-Anger expansion and their probability mixture, $ P_{erf,\gamma ,p,K_1,K_2}(\lambda) = pP_{erf,\gamma ,K_1}(\lambda) + (1-p)P_{erf,\gamma ,K_2}(\lambda), p \in [0,1)$ approximate the error function $erf(\gamma \lambda) = \frac{2}{\pi}\int^{\gamma \lambda}_0 e^{-t^2}dt$ with truncation error bounded by $a_1^\prime$, $a_2^\prime$ and $b^\prime$ respectively.
    \begin{equation}
        \begin{aligned}
            a_1^\prime &=  \frac{\gamma e^{-\gamma ^2/2}}{\sqrt{\pi}} \frac{4(\gamma^2/2)^{(K_1+1)/2}}{2^{(K_1+1)/2}((K_1+1)/2)!}\\
            a_2^\prime &=  \frac{p}{1-p}\frac{\gamma e^{-\gamma ^2/2}}{\sqrt{\pi}} \frac{4(\gamma^2/2)^{(K_1+1)/2}}{2^{(K_1+1)/2}((K_1+1)/2)!}\\ 
            b^\prime &= \frac{\gamma e^{-\gamma ^2/2}}{\sqrt{\pi}} \frac{4(\gamma^2/2)^{(K_2+1)/2}}{2^{(K_2+1)/2}((K_2+1)/2)!}
        \end{aligned}
    \end{equation}
\end{lemma}
\begin{proof}
Calculate that 
    \begin{enumerate}
    \item $|erf(\lambda) -  P_{erf,\gamma ,K_1}(\lambda)| \le a_1^\prime$
    \begin{equation}
        \begin{aligned}
            \epsilon_{erf,\gamma,K_1} &= |erf(\lambda) -  P_{erf,\gamma ,K_1}(\lambda)| \\
            &= \left|\frac{2\gamma e^{-\gamma ^2/2}}{\sqrt{\pi}}\sum^{\infty}_{k=(K_1+1)/2}J_k\left(\frac{\gamma ^2}{2}\right)(-1)^k\left(\frac{T_{2k+1}(\lambda)}{2k+1}-\frac{T_{2k-1}(\lambda)}{2k-1}\right)\right| \\
            &\le \frac{2\gamma e^{-\gamma ^2/2}}{\sqrt{\pi}}\sum^{\infty}_{k=(K_1+1)/2}\left|J_k\left(\frac{\gamma ^2}{2}\right)\right|\left|\left(\frac{1}{2k+1}-\frac{1}{2k-1}\right)\right| \\
            &\le \frac{2\gamma e^{-\gamma ^2/2}}{\sqrt{\pi}}\sum^{\infty}_{k=(K_1+1)/2}\left|J_k\left(\frac{\gamma ^2}{2}\right)\right| \\
            &\le \frac{\gamma e^{-\gamma ^2/2}}{\sqrt{\pi}} \frac{4(\gamma^2/2)^{(K_1+1)/2}}{2^{(K_1+1)/2}((K_1+1)/2)!} =: a_1^\prime
        \end{aligned}
    \end{equation}

    \item $|erf(\lambda) -  P_{erf,\gamma ,K_2}(\lambda)| \le a_2^\prime$
    \begin{equation}
        \begin{aligned}
            \epsilon_{erf,\gamma,K_2} &= |erf(\lambda) -  P_{erf,\gamma ,K_2}(\lambda)| \\
            &= \left|\frac{2\gamma e^{-\gamma ^2/2}}{\sqrt{\pi}}\left(-\frac{p}{1-p}\sum^{(K_2-1)/2}_{k=(K_1+1)/2}J_k\left(\frac{\gamma ^2}{2}\right)(-1)^k\left(\frac{T_{2k+1}(\lambda)}{2k+1}-\frac{T_{2k-1}(\lambda)}{2k-1} \right)\right.\right. \\
            &\quad +\left.\left. \sum^{\infty}_{k=(K_2+1)/2}J_k\left(\frac{\gamma ^2}{2}\right)(-1)^k\left(\frac{T_{2k+1}(\lambda)}{2k+1}-\frac{T_{2k-1}(\lambda)}{2k-1}\right)\right)\right| \\
            &\le \left|\frac{2\gamma e^{-\gamma ^2/2}}{\sqrt{\pi}}\left(-\frac{p}{1-p}\sum^{\infty}_{k=(K_1+1)/2}J_k\left(\frac{\gamma ^2}{2}\right)(-1)^k\left(\frac{T_{2k+1}(\lambda)}{2k+1}-\frac{T_{2k-1}(\lambda)}{2k-1} \right)\right.\right. \\
            &\quad +\left.\left. \left(1+\frac{p}{1-p}\right)\sum^{\infty}_{k=(K_2+1)/2}J_k\left(\frac{\gamma ^2}{2}\right)(-1)^k\left(\frac{T_{2k+1}(\lambda)}{2k+1}-\frac{T_{2k-1}(\lambda)}{2k-1}\right)\right)\right| \\
            &\le \frac{p}{1-p}\frac{2\gamma e^{-\gamma ^2/2}}{\sqrt{\pi}}\sum^{\infty}_{k=(K_1+1)/2}\left|J_k\left(\frac{\gamma ^2}{2}\right)\right| \\
            &\le \frac{p}{1-p}\frac{\gamma e^{-\gamma ^2/2}}{\sqrt{\pi}} \frac{4(\gamma^2/2)^{(K_1+1)/2}}{2^{(K_1+1)/2}((K_1+1)/2)!} := a_2^\prime
        \end{aligned}
    \end{equation}

    \item $|erf(\lambda) -  P_{erf,\gamma ,p,K_1,K_2}(\lambda)| \le b^\prime$
        \begin{equation}
        \begin{aligned}
            \epsilon_{erf,\gamma ,p,K_1,K_2}(\lambda) &= |erf(\lambda) -  P_{erf,\gamma ,p,K_1,K_2}(\lambda)| \\
            &= \left|\frac{2\gamma e^{-\gamma ^2/2}}{\sqrt{\pi}}\sum^{\infty}_{k=(K_2+1)/2}J_k\left(\frac{\gamma ^2}{2}\right)(-1)^k\left(\frac{T_{2k+1}(\lambda)}{2k+1}-\frac{T_{2k-1}(\lambda)}{2k-1}\right)\right| \\
            &\le \frac{2\gamma e^{-\gamma ^2/2}}{\sqrt{\pi}}\sum^{\infty}_{k=(K_2+1)/2}\left|J_k\left(\frac{\gamma ^2}{2}\right)\right|\left(\frac{1}{2k+1}-\frac{1}{2k-1}\right) \\
            &\le \frac{2\gamma e^{-\gamma ^2/2}}{\sqrt{\pi}}\sum^{\infty}_{k=(K_2+1)/2}\left|J_k\left(\frac{\gamma ^2}{2}\right)\right| \\
            &\le \frac{\gamma e^{-\gamma ^2/2}}{\sqrt{\pi}} \frac{4(\gamma^2/2)^{(K_2+1)/2}}{2^{(K_2+1)/2}((K_2+1)/2)!} =: b^\prime
        \end{aligned}
        \end{equation}
\end{enumerate}
\end{proof}

We can then upper bound the error of approximating the truncated linear function 

\begin{lemma}{(Truncation error of approximating linear function)}

    Polynomial functions $\hat{P}_{\Gamma,\delta,K_{1(2)}}(\lambda)$ in Eq.~\eqref{ep:QSP_lin_v_1(2)} and the probability mixture $ \hat{P}_{\Gamma,\delta,p,K_1,K_2}(\lambda) = p\hat{P}_{\Gamma,\delta,K_{1}}(\lambda) + (1-p)\hat{P}_{\Gamma,\delta,K_{2}}(\lambda), p \in [0,1)$ approximate the truncated linear function $f_{\Gamma,\delta}(\lambda) = \lambda/(2\Gamma), \quad |\lambda| \in [0,\Gamma]$ with truncation error bounded by $a_1$, $a_2$ and $b$ respectively.
    \begin{equation}
        \begin{aligned}
            a_1 &=  \frac{8\Gamma e^{-8\Gamma ^2}}{\sqrt{\pi}} \frac{4(8\Gamma^2)^{K_1/2}}{2^{K_1/2}(K_1/2)!}\\
            a_2 &=  \frac{p}{1-p}\frac{8\Gamma e^{-8\Gamma ^2}}{\sqrt{\pi}} \frac{4(8\Gamma^2)^{K_1/2}}{2^{K_1/2}(K_1/2)!} = \frac{p}{1-p}a_1\\ 
            b &= \frac{8\Gamma e^{-8\Gamma ^2}}{\sqrt{\pi}} \frac{4(8\Gamma^2)^{K_2/2}}{2^{K_2/2}(K_2/2)!}
        \end{aligned}
    \end{equation}
\end{lemma}
\begin{proof}
    This is followed by substituting Eq.~\eqref{eq:QSP_lin_erf_bound_transfer} into lemma \ref{le:QSP_erf_error}.
\end{proof}

\subsection{Proof of Corollary~\ref{cor:USA}}

\begin{proof}

With two classically computed $\vec{\varphi}_1 \in \mathbb{R}^{2K_1+1}, \vec{\varphi}_2 \in \mathbb{R}^{2K_2+1}$ , we can implement $\hat{W}_{\vec{\varphi}_{1(2)}}$ and $\hat{W}_{\vec{\varphi}_{m}} = p\hat{W}_{\vec{\varphi}_{1}} + (1-p)\hat{W}_{\vec{\varphi}_{2}}$ such that they approximate an unitary $U$ implementing truncated linear amplification by bounded errors
\begin{equation}
    \begin{aligned}
        \|\bra{0}_c\bra{0}_b\bra{G}_a\hat{W}_{\vec{\varphi}_{1}}\ket{G}_a\ket{0}_b\ket{0}_c - U\| \le a_1 \\
        \|\bra{0}_c\bra{0}_b\bra{G}_a\hat{W}_{\vec{\varphi}_{2}}\ket{G}_a\ket{0}_b\ket{0}_c - U\| \le a_2 \\
        \|\bra{0}_c\bra{0}_b\bra{G}_a\hat{W}_{\vec{\varphi}_{m}}\ket{G}_a\ket{0}_b\ket{0}_c - U\| \le b
    \end{aligned}
\end{equation}

The operator norm further bound the state distance after quantum channels since $\|V\rho V^\dagger - U\rho U^\dagger\| \le \|V-U\|, \, \forall \rho $. Employing lemma \ref{le:mix} gives the final result since $b$ is exponentially smaller than $a_{1(2)}$, i.e.
\begin{equation}
\begin{aligned}
    \left\|\mathcal{V}_{\text{mix}}(\rho) - U\rho U^\dagger\right\| &\le 4b+2p a_1^2 + 2(1-p)a_2^2 \\
    & \le \max\left\{8b,\frac{4}{1-p}a_1^2\right\}
\end{aligned}
\end{equation}
\end{proof}

\end{document}